\documentclass[sigconf]{acmart}
\usepackage{subcaption}
\usepackage{graphicx}
\usepackage{float}
\usepackage{xcolor}
\usepackage{xspace}
\usepackage{environ}
\usepackage{framed}
\usepackage{multirow} 
\usepackage{algorithm}
\usepackage{fancyvrb}      
\usepackage[shortlabels]{enumitem}

\setlist[enumerate]{leftmargin=0.7cm,topsep=0.5mm}
\setlist[itemize]{leftmargin=0.5cm,topsep=0.5mm}

\definecolor{RedColor}{HTML}{e66101}
\definecolor{OrangeColor}{HTML}{fdb863}
\definecolor{YellowColor}{HTML}{fec44f}
\definecolor{BlueColor}{HTML}{b2abd2}
\definecolor{PurpleColor}{HTML}{5e3c99}
\definecolor{Ours0Color}{HTML}{ABDDA4}
\definecolor{Ours16Color}{HTML}{72C166}
\definecolor{Ours32Color}{HTML}{38A528}
\definecolor{PrestoColor}{HTML}{999999}
\definecolor{PostgresColor}{HTML}{6D6D6D}
\definecolor{GreenColor}{HTML}{38A528}
\definecolor{YellowColor}{HTML}{ffb570}
\definecolor{BlueColor}{HTML}{7081ff}
\definecolor{PinkColor}{HTML}{ffb0c2}
\definecolor{ComputeColor}{HTML}{ffb0c2}
\definecolor{ReadColor}{HTML}{cf3457}
\definecolor{WriteColor}{HTML}{ffb570}
\definecolor{vintagegreen}{HTML}{ABDDA4}
\definecolor{OursColor}{HTML}{38A528}
\definecolor{GreedyColor}{HTML}{7081ff}
\definecolor{RandomColor}{HTML}{ffb570}
\definecolor{NoneColor}{HTML}{6D6D6D}
\definecolor{Redborder}{HTML}{805861}
\definecolor{Greenborder}{HTML}{384180}
\definecolor{Blueborder}{HTML}{566F52}
\definecolor{Greyborder}{HTML}{4D4D4D}
\definecolor{Lightgrey}{HTML}{dadada}
\definecolor{ExampleColor1}{HTML}{7081ff}
\definecolor{ExampleColor2}{HTML}{ffb0c2}
\definecolor{Lightred}{HTML}{ffb09c}
\definecolor{Lightblue}{HTML}{b8e2f2}
\definecolor{FlagColor}{HTML}{CCCCCC}
\definecolor{vintageblue}{HTML}{7081ff}
\definecolor{vintagered}{HTML}{ffb0c2}
\definecolor{NoOptColor}{HTML}{264653}
\definecolor{LRUColor}{HTML}{777777}
\definecolor{RandomColor}{HTML}{2a9d8f}
\definecolor{GreedyColor}{HTML}{e9c46a}
\definecolor{HeuristicColor}{HTML}{f4a261}
\definecolor{SCColor}{HTML}{e76f51}
\definecolor{AllColor}{HTML}{CCCCCC}
\definecolor{SAColor}{HTML}{ffb0c2}
\definecolor{SeparatorColor}{HTML}{9b5de5}
\definecolor{BlueColor}{HTML}{0081a7}
\definecolor{cAmain}{HTML}{e66101}
\colorlet{cAlight}{cAmain!25}
\colorlet{cAlightlight}{cAmain!5}
\definecolor{cBmain}{HTML}{5e3c99}
\colorlet{cBlight}{cBmain!25}
\definecolor{cCmain}{HTML}{b2abd2}
\colorlet{cClight}{cCmain!25}
\definecolor{cDmain}{HTML}{fdb863}
\colorlet{cDlight}{cDmain!25}
\definecolor{cZmain}{HTML}{030303}  
\colorlet{cZlight}{cZmain!25}
\colorlet{cZlightlight}{cZmain!5}
\colorlet{cZlightlightlight}{cZmain!1}
\definecolor{cPositivemain}{HTML}{2ca02c}  
\colorlet{cPositivelight}{cPositivemain!25}
\definecolor{cNegativemain}{HTML}{d62728}  
\colorlet{cNegativelight}{cNegativemain!25}

\newcommand\system{Kishuboard\xspace}

\newcommand{\CellWithForcedBreak}[2][c]{%
    \begin{tabular}[#1]{@{}c@{}}#2\end{tabular}%
}
\usepackage{multirow}  
\usepackage{algorithm}
\usepackage{xcolor}      
\usepackage{fancyvrb}

\newcommand{\fix}[1]{{#1}}

\makeatletter
\apptocmd\@maketitle{{\myfigure{}\par}}{}{}
\makeatother

\copyrightyear{2025}
\acmYear{2025}
\setcopyright{cc}
\setcctype{by}
\acmConference[CHI '25]{CHI Conference on Human Factors in Computing Systems}{April 26-May 1, 2025}{Yokohama, Japan}
\acmBooktitle{CHI Conference on Human Factors in Computing Systems (CHI '25), April 26-May 1, 2025, Yokohama, Japan}\acmDOI{10.1145/3706598.3714141}
\acmISBN{979-8-4007-1394-1/2025/04}

\begin{document}

\title{Enhancing Computational Notebooks with Code+Data Space Versioning}

\author{Hanxi Fang}
\affiliation{%
  \institution{University of Illinois Urbana-Champaign}
  \city{Urbana}
  \state{Illinois}
  \country{USA}
}
\email{hanxif2@illinois.edu}

\author{Supawit Chockchowwat}
\affiliation{%
  \institution{University of Illinois Urbana-Champaign}
  \city{Urbana}
  \state{Illinois}
  \country{USA}
}
\email{supawit2@illinois.edu}

\author{Hari Sundaram}
\affiliation{%
  \institution{University of Illinois Urbana-Champaign}
  \city{Urbana}
  \state{Illinois}
  \country{USA}
}
\email{hs1@illinois.edu}

\author{Yongjoo Park}
\affiliation{%
  \institution{University of Illinois Urbana-Champaign}
  \city{Urbana}
  \state{Illinois}
  \country{USA}
}
\email{yongjoo@illinois.edu}

\begin{abstract}



There is a gap between how people explore data and how Jupyter-like computational notebooks are designed. 
People explore data nonlinearly, using execution undos, branching, and/or complete reverts, whereas notebooks are designed for sequential exploration. 
Recent works like ForkIt are still insufficient 
    to support these multiple modes of nonlinear exploration in a unified way.

In this work,
    we address the challenge by introducing 
        two-dimensional code+data space versioning for computational notebooks
    and verifying its effectiveness using our prototype system, \system,
        which integrates with Jupyter.
By adjusting code and data knobs,
    users of \system can intuitively manage 
        the state of computational notebooks in a flexible way,
    thereby achieving both execution rollbacks and checkouts
        across complex multi-branch exploration history.
Moreover, this two-dimensional versioning mechanism
    can easily be presented along with a friendly one-dimensional history.
Human subject studies indicate that
    \system significantly enhances
        user productivity in various data science tasks.


\end{abstract}


\maketitle


\section{Introduction}
\label{sec:intro}


\begin{figure*}[t]
    \includegraphics[width=0.75\linewidth]{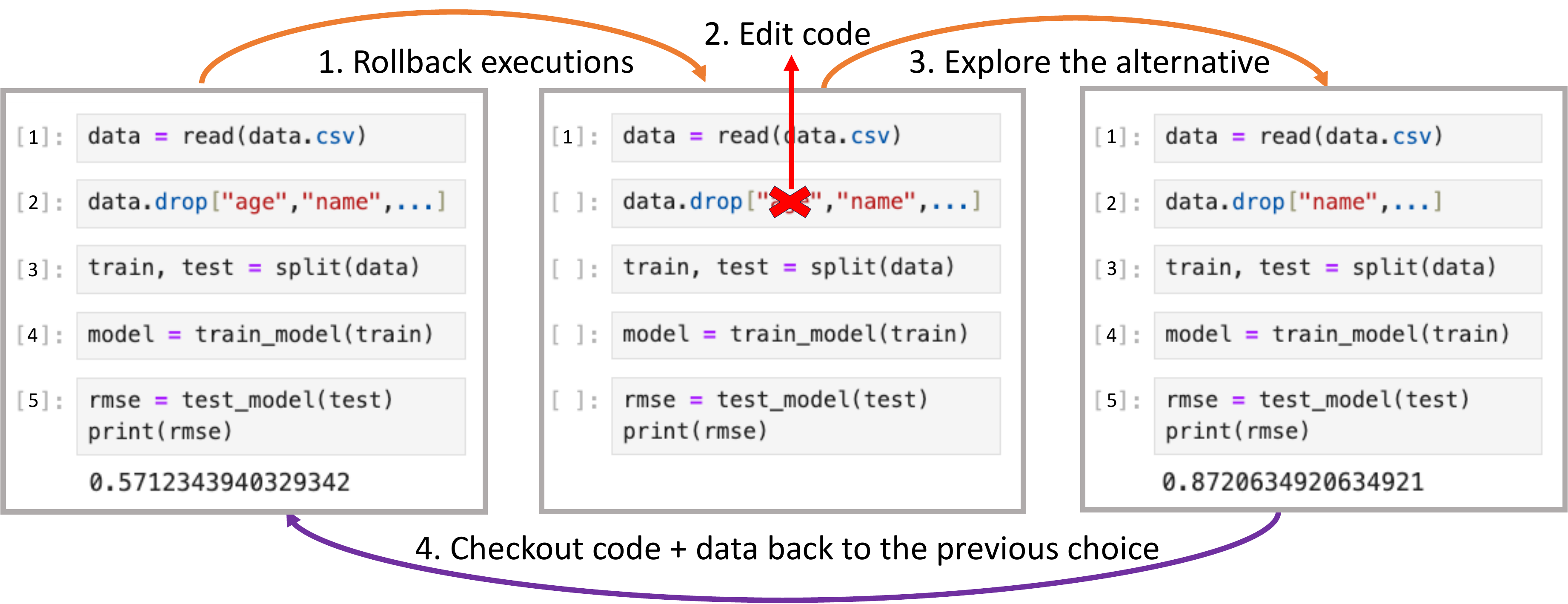}
    \vspace{-2mm}
    \Description{The figure illustrates the motivation behind Code+Data space versioning in data science workflows using a series of code cells in a notebook interface.
The first panel shows a sequence of five code cells, with each cell being executed and having an execution counter beside it. The counter is from 1 to 5. 
Then, an orange arrow connects the first panel to the middle panel. The arrow has the text "rollback executions" on it.
The middle panel represents a rolled-back version of the first panel, where the code is still there, but all the execution counters from 2 to 4 are wiped, meaning the executions are rolled back. Compared to the left panel, an edit is made in the second code cell.
The right panel shows the results after re-execution following the code edit, reflecting another exploration branch.
Finally, there's an arrow connecting the right panel back to the left one, with the text "Checkout Code + Data Back to the Previous Choice" on the arrow. : This demonstrates that data scientists "checkout" back to the previous exploration path after exploring the alternative.
The goal is to show diverse types of modifications in code and data states (rollback execution and checkout).}
    \caption{\fix{The motivation for code+data space versioning. 
    Data scientists often wish to undo (only) executions while keeping the code
        (i.e., \textbf{execution rollback}; left to middle)
        for testing alternative methods (e.g., edits in the middle followed by executions in the right);
    Or they wish to completely revert all of their activities,
        jumping back to a certain point in the past (i.e., \textbf{checkout}; right to left).
    Our goal is to enable these various types of code/data state modifications
        through an intuitive user interface.
    }}
    \label{fig:intro:motivation}
\end{figure*}


There is a significant gap 
    between how people explore data 
    and how computational notebooks (e.g., Jupyter~\cite{jupyter}, Colab~\cite{colab})
        are designed.
This discrepancy
    causes pain and wasted resources.
People explore data \emph{nonlinearly},
    while computational notebooks are designed for
        \emph{sequential} exploration only.
In Fig~\ref{fig:intro:motivation}, for example, 
Alice evaluates a model after dropping the \texttt{age} feature from the data (left).
For higher accuracy, she wishes to test another feature set.
\emph{If allowed,}
    she could \textbf{roll back} some of the executions (left to middle),
        edit the feature selection (middle),
        and re-run the training logic (right).
The new model shows lower accuracy.
Alice would end this alternative exploration
    by \textbf{checking out} both code and data\footnote{Checking out “data” refers to retrieving and restoring a previous version of runtime variable names, values, and execution history. Shortly, \S\ref{sec:version-set} defines ``data'' formally.} from an earlier version (right to left).
Unfortunately, none of the existing off-the-shelf computational notebooks---from 
    open-source tools (e.g., Jupyter~\cite{jupyter}, R Markdown~\cite{rmarkdown})
        to commercial services 
            (e.g., Google Colab~\cite{colab}, 
    Mode Analytics~\cite{modeanalytics})---allow
    this nonlinear workflow.
If they offer nonlinear exploration,
    it would completely reshape how data scientists
        explore data,
    while immediately addressing 
        the pressing needs
        shared by practitioners (\S\ref{sec:interview}).

Novel interfaces and systems are proposed to narrow this gap; however, they remain insufficient.
For example, ForkIt~\cite{weinman2021fork} allows for
    pursuing multiple exploration paths simultaneously with side-by-side columns; each column represents an independent path branched out
        from a common ancestor.
While it is a significant improvement over sequential exploration,
    ForkIt requires \emph{proactive} planning without the ability to 
        \emph{retrospectively} branch out and return.
    Moreover, it is unclear how users can manage more than a handful of exploration paths,
        each of which may recursively branch into more alternatives.
Approaches like Diff-in-the-Loop~\cite{wang2022diff}, \fix{Loops~\cite{eckelt2024loops}, and} Chameleon~\cite{hohman2020understanding}
    offer visual aids for comparing data between iterative explorations.
While these visualization techniques are helpful,
    they can be more successful when the underlying systems,
        like computational notebooks,
    are designed to support iterations through nonlinear exploration,
        which is lacking today.
Recording code execution order~\cite{head2019managing,kery2019towards,kery2018interactions,gitbook}
and highlighting branched code structures \fix{by offering non-linear visual cell layout~\cite{ramasamy2023visualising,subramanian2019supporting,venkatesan2022automatic,wang2022stickyland,kery2017variolite,harden2022exploring, harden2023there}}
    are also beneficial,
but they cannot manipulate 
    the data state of computational notebooks,
which is required for the workflow 
    described in Fig~\ref{fig:intro:motivation}.

In this work, we propose two-dimensional code+data space versioning for computational notebooks,
    allowing novel primitives,
        \textbf{execution rollback} and \textbf{code+data checkout},
    for \emph{consistently} navigating past states.
In Fig~\ref{fig:intro:2D_checkout}, for example, 
    users can \textbf{roll back} executions by moving the ``Head: Variable'' box
    from the latest commit (i.e., \texttt{rmse = test\_model}
\texttt{(test)} on top)
        to the first commit (i.e., \texttt{data=}
    \texttt{read(data.csv)}) at bottom).
Since only the data (i.e., variable) state has changed, users can still
    see all the code blocks they wrote before.
Alternatively,
    they can check out both code and data,
    moving both ``Head: Variable'' and ``Head: Code'',
        thereby effectively time-traveling back to the target commit.

These two primitives, rollback and checkout, ensure \emph{consistency};
    that is, the retained/restored data is equivalent to 
        what would have been produced 
            if (some of) the code had been executed from scratch.
We formalize this notion in \S\ref{sec:version}.
Our approach is unique
    compared to the one-dimensional versioning
        employed by code version control (e.g., git),
            database checkpointing, etc.
While this paper focuses on computational notebooks for clearer presentation,
    the same design principle applies to
        a broader range of interactive data science,
    including Matlab, Stata, SAS, and so on.

To verify the effectiveness of our code+data versioning,
we have built an end-to-end prototype, called \system.
Any code+data state manipulation on \system 
    alters in real-time the actual objects and code served on Jupyter.
Likewise, code executions in Jupyter
    appear immediately on \system as additional commits,
    advancing ``Head: Code'' and ``Head: Variable'' appropriately
        with optional branching if necessary.
This seamless integration is enabled by our novel component 
    that monitors Jupyter and 
        synchronizes its state according to user actions in real-time (\S\ref{sec:design}).

We evaluate the usefulness of nonlinear data science
    using Kishu-board through human subject studies (\S\ref{sec:evalmethod}).
Our quantitative study shows that \system helps users improve their productivity, especially in intensive workloads when re-execution takes a long time. Our qualitative data shows that users think \system's features are useful, the UI designs are mostly user-friendly, and they're willing to use or even pay for \system in the future (\S\ref{sec:eval}).


\begin{figure}[t]
\centering
\begin{subfigure}[b]{0.9\linewidth}
    \includegraphics[width=0.91\linewidth]{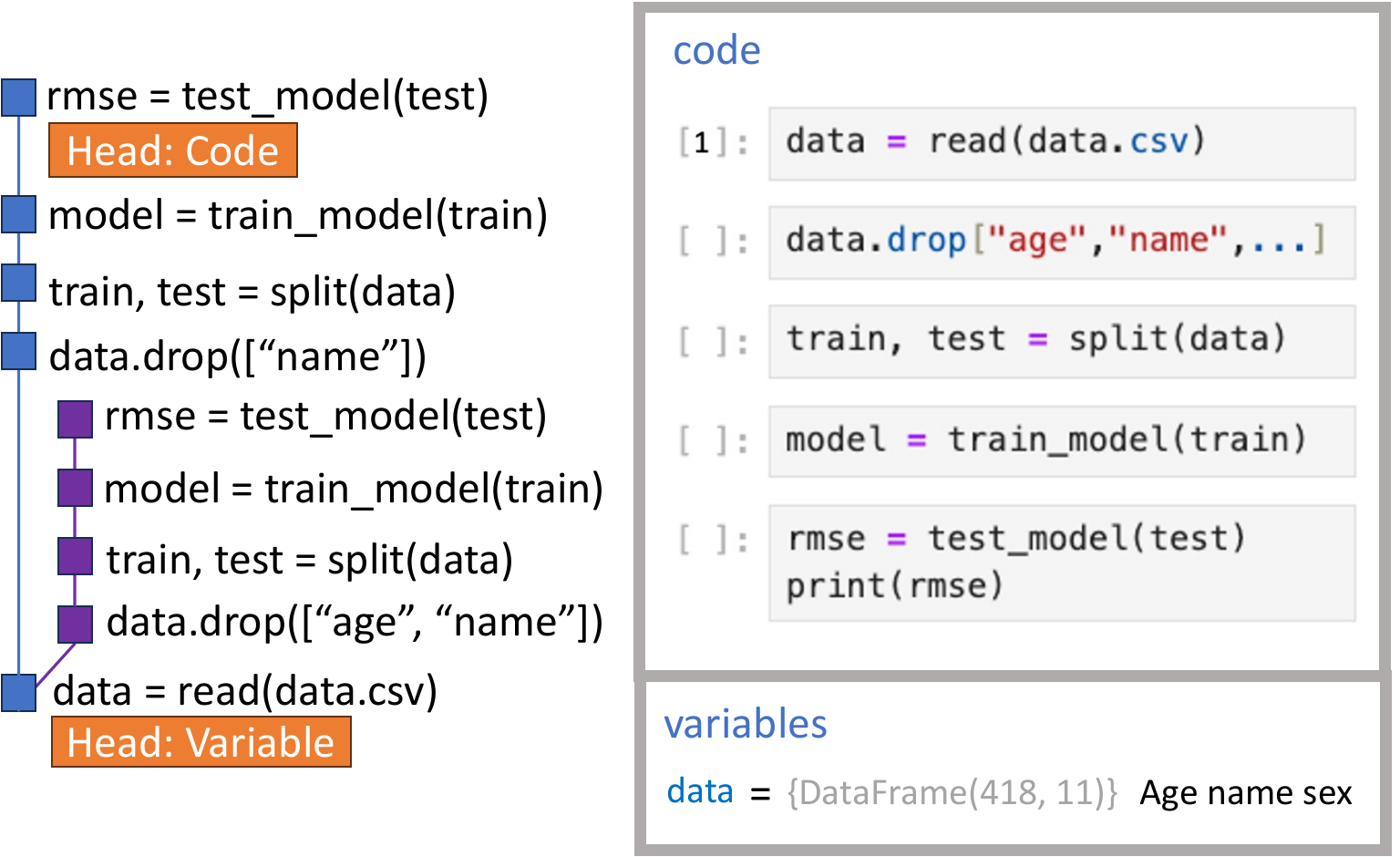}
    \vspace{-2mm}
    \caption{Code and data states for \textbf{execution rollback}}
    \vspace{3mm}
\end{subfigure}
\begin{subfigure}[b]{0.9\linewidth}
    \includegraphics[width=0.91\linewidth]{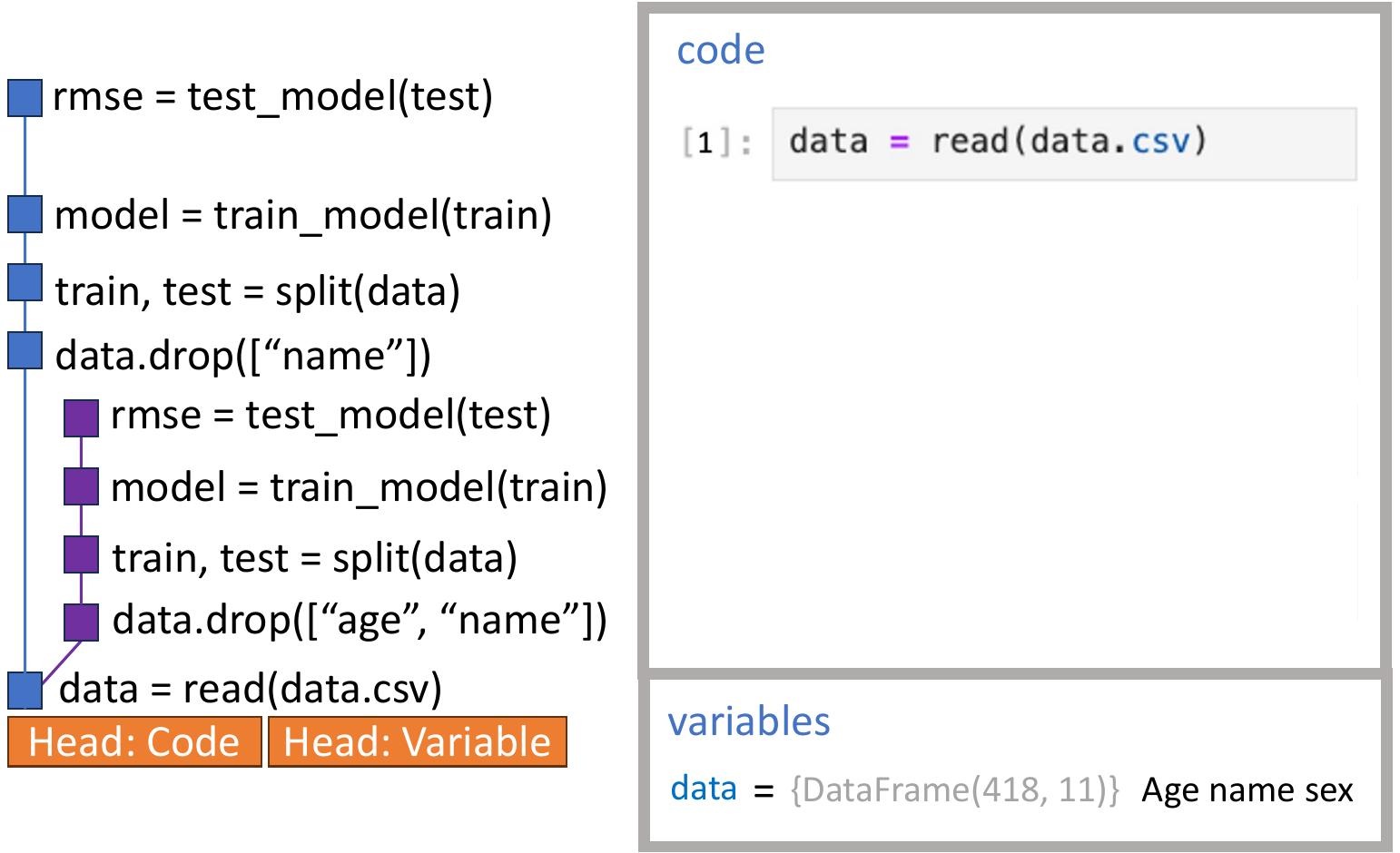}
    \vspace{-2mm}
    \caption{Code and data states for \textbf{checkout}}
\end{subfigure}

\vspace{-2mm}
\Description{The figure illustrates two scenarios in managing code and data states within a notebook environment: execution rollback and code+data checkout.
Panel (a): Code and Data States for Execution Rollback:
The left side shows a series of code cells, each representing a step in a data science workflow (reading data, dropping columns, splitting data, training a model, and evaluating the model).
Two branches are shown:
The main branch (blue boxes) follows the original execution path.
A secondary branch (purple boxes) represents a rollback to an earlier data state, maintaining the same code but with altered data (e.g., dropping different columns).
The "Head: Code" and "Head: Variable" labels indicate the points where the code and data states are currently focused.
Panel (b): Code and Data States for Checkout:
The right side shows the result of performing a checkout operation, where both code and data states are reverted to a specific past point.
The variables section at the bottom indicates that the data is now a different version of the DataFrame, reflecting changes from the original state}
\caption{\fix{By adjusting two knobs (i.e., orange boxes for code and data),
    we can achieve (a) execution rollback and (b) code+data checkout.
    \textbf{Execution rollback} is equivalent to altering only the data
        while keeping the code the same.
    \textbf{Checkout} is equivalent to altering both code/data
        to exactly a point that existed in the past.
    In addition to proposing this two-dimensional versioning,
        we formalize \emph{consistent} code/data states (\S\ref{sec:version}),
        develop a working prototype (\S\ref{sec:design} and \S\ref{sec:others}),
            and evaluate its usefulness in data science tasks 
            (\S\ref{sec:evalmethod} and \S\ref{sec:eval}).}}
\label{fig:intro:2D_checkout}
\end{figure}


This work makes the following contributions:
\begin{enumerate}
    \item \textbf{Conceptualizing Two-Dimensional Code+Data Versioning.} 
    To our knowledge, we are the first to propose and implement a new versioning system for interactive data science (\system) that supports non-linear exploration of both code and data. In contrast, prior work only supports the versioning for code~\cite{head2019managing,kery2019towards,kery2018interactions,gitbook}. We show through formative interviews with data scientists the critical need for versioning both code and data states. Using those insights,  we introduce the two new primitives: code+data checkouts, and code execution rollbacks. We build a real-time versioning system that supports these two primitives and integrates with a computational notebook. 
    The system ensures the code and data state consistency and supports the temporal exploration of a two-dimensional code+data evolution history (\S\ref{sec:design}, \S\ref{sec:version}, and Appendix~\ref{sec:others}). 
    Supporting non-linear exploration significantly 
        enhances data science workflows.
    \item \textbf{Interface for Two-Dimensional Code+Data Exploration.} 
    Prior interfaces have focused on navigating changes to code~\cite{head2019managing,kery2019towards,kery2018interactions,gitbook},  pain point of notebooks~\cite{ramasamy2023visualising,subramanian2019supporting,venkatesan2022automatic,wang2022stickyland,kery2017variolite}, visualizing the change of data frames~\cite{hohman2020understanding,wang2022diff}, or user behaviors when exploring data~\cite{kery2018story,guo2013workflow,raghunandan2023code,hohman2020understanding,liu2019understanding}. 
    We extend these existing interfaces to
        support code-data space navigation.
    Specifically, our key design elements 
    include mapping the two-dimensional code+data history to a one-dimensional graph,
        allowing
    users to select a previous code/data version,
        inspect its state, and load it.
    Additional tagging and searching 
    facilitate locating desirable commits. 
    Human subject studies demonstrate time savings; 
        this productivity benefit increases for more time-consuming workflows.  
\end{enumerate}



\section{Background and Why Now}
\label{sec:background}

This section briefly describes the advances
    in data science checkpointing.
We use Jupyter~\cite{jupyter} as an example
    considering its impact~\cite{colab,paperspace,modeanalytics}. 
This will help readers
    understand existing pain points (\S\ref{sec:interview})
    and our system design (\S\ref{sec:design} and
        \S\ref{sec:version}).


\subsection{Existing System for Interactive Data Science}

Jupyter is a web interface
    running on top of a Python interpreter.
Jupyter lets users submit a code block one at a time.
The result can be presented in diverse formats,
    ranging from text to JavaScript-based interactive plots.
Leveraging the Python ecosystem,
    users can import various libraries
        for data analytics, machine learning, visualization, and so on.
In this work,
    we focus on providing 
    nonlinear data exploration capabilities
    for this general framework
    without requiring any changes to
        Jupyter or any other libraries.

\subsection{No Checkpointing in Existing Data Science}

Database systems 
    such as PostgreSQL and Microsoft SQL Server create \emph{checkpoints},
    periodically saving changes~\cite{aries, sears2009segment}.
In contrast, data science systems lack mechanisms for
    identifying changes of kernel data,
        an important premise of checkpointing.
Existing databases achieve this through centralized buffer pages.
If data is updated,
    the corresponding pages are marked \emph{dirty} and saved later.
Instead, data science systems omit central data spaces and allow individual libraries 
    to manage data
        using
        shared memory~\cite{arrow}, 
        GPUs~\cite{pytorch,tensorflow}, 
        and remote machines~\cite{spark2012,ray}
    for high performance.
These libraries rarely trace changes in data.
As a result,
    nearly all data science systems,
        including Jupyter,
    lacks checkpointing.

\subsection{Advances in Data Science Checkpointing}
\label{sec:background-advances}

There are recent advances toward
    checkpointing data science systems (i.e., checkpointing kernel data)~\cite{chex,elasticnotebook,elasticnotebookdemo,kishu}.
The observation is that the standard serialization protocol 
    is insufficient for capturing the state.
Instead, object dependencies must also be considered
    to capture the state correctly~\cite{elasticnotebook}. Once dirty objects are identified, incremental checkpointing~\cite{kishu} saves only the relevant variables and tracks their changes over time, reusing unchanged variable states from earlier checkpoints. Incremental checkout~\cite{kishu}, on the other hand, efficiently transitions the current kernel state to the desired state. 
The idea has been validated for a large number (146) of data science classes~\cite{kishu} including
    in-memory analytics~\cite{harris2020array,mckinney2011pandas},
    GPU training~\cite{pytorch,tensorflow,hfpipeline},
    and distributed computing~\cite{spark2012,ray}.
The method has been integrated 
    with batch jobs~\cite{chockchowwat2023transactional}
    and Jupyter~\cite{elasticnotebookdemo}.

However, the ability to checkpoint data
    is insufficient for nonlinear data science.
First, we need a principled design
    for managing code and data states consistently
        and presenting them clearly.
Second, we should validate the usefulness
    of nonlinear data exploration for actual data science tasks.
This paper aims to close this gap.



\section{Formative Interview}
\label{sec:interview}

We met with six data scientists
    to hear their thoughts on nonlinear data science and how checkpoints for both code and data can be potentially used to support it.
These conversations suggest a new versioning mechanism
    in the two-dimensional code+data space. Compared to existing works that focus only on the pain point of notebooks~\cite{ramasamy2023visualising,subramanian2019supporting,venkatesan2022automatic,wang2022stickyland,kery2017variolite} or the user's behavior when exploring data~\cite{kery2018story,guo2013workflow,raghunandan2023code,hohman2020understanding,liu2019understanding}, we also learned lessons about how versioning for both code and data can potentially help them and what's the concerns when it comes to developing a complete version management system for nonlinear interactive data science based on those code and data checkpoints (\S\ref{sec:interview_checkpointing}).
    


\begin{table}[t]
\caption{Pain points in Jupyter}
\label{tab:painpoints}

\vspace{-2mm}
\centering

\small
\begin{tabular}{lll}
\toprule
\textbf{Pain Point} & \textbf{Interviewee} & \textbf{Do We Solve}
    \\ \midrule
P1: Session suspension deletes data  & F5, F6 & Yes \\
P2: Cannot restore overwritten data & F2, F4 & Yes \\
P3: Hard to debug & F3 & Yes \\
P4: No data versioning & F3 & Yes \\
P5: History becomes messy quickly  & F1 & Partly \\
P6: Hard to scale resources & F4 & Partly \\
\bottomrule
\end{tabular}
\end{table}



\begin{table*}[t]

\caption{Suggestions for nonlinear data science}
\label{tab:suggestions}
\vspace{-2mm}

\centering
\small

\begin{tabular}{ l l l }
\toprule
\textbf{Suggestion} & \textbf{Interviewees}  & \textbf{Our Approach}  
\\ \midrule
S1: Keep code and revert only variable state& F1, F4&    
\multirow{2}{*}{$\left.\begin{array}{l}
                \\
                \\
        \end{array}\right\rbrace$ Checkout code and kernel states separately}

    \\
S2: Distinguish the changes in data sources and code & F2 &\\
S3: Show code/variable change history separately 
    & F2, F4, F5 
    & 1-D timeline for 2-D state transition \\
S4: Concern: Too many commits for one notebook & F1, F4 & 
    Context-aware auto-folding (of commit history)
    \\
S5: Show the lineage of variables
    & F1, F2, F3 
    & Search: highlight commits that modified the variable \\
S6: Compare two commits & F2 
    & Independent \texttt{diff} of code and variable \\
\midrule
S7: Different commit granularity for different use cases& F2 & 
    \multirow{3}{22em}{Can be supported, but not implemented in this work }\\
S8: Cherry picking variable from a different branch & F1, F6 & \\
S9: Sharing, merging, collaboration& F1, F2, F4, F5 &\\
\bottomrule
\end{tabular}
\end{table*}


\subsection{Interview Methodology}

\paragraph{Group}

Five data scientists (F1–F5) were from Meta, Microsoft, and start-ups.
They all had PhD in computer science and 
    more than two years of industry experience. 
We also interviewed with a fourth-year PhD student (F6)
    studying machine learning.

\paragraph{Procedure}

We had a 60-min semi-structured interview with every data scientist. Each interview had three parts. 
First, we discussed how they used Jupyter. 
Second, we asked about current pain points. 
Finally, we shared the idea of using code+data checkpoints to support nonlinear exploration and 
    asked how it could be enhanced. 
\fix{During the process, we each noted the interviewees' answers. Later, we combined our documents and summarized their answers into several categories.}


\subsection{How They Use Jupyter}

They used Jupyter for three main purposes.
The first was exploratory modeling. 
Using a sample,
they would experiment with different parameters (F6), datasets (F4, F6), or models like gradient boosting (F2, F5). 
Due to resource constraints, final models were usually trained outside Jupyter (F1-F6).
The second use case was data loading and transformation.
F4 would start a new session, load a dataset, experiment with various data transformation methods, 
and then save the transformed results to CSV for future use.
Finally, Jupyter was used for quick debugging and sanity checks.
F3 would copy suspicious code into Jupyter
    to verify data formats.
In all these, Jupyter was a convenient tool for examining
    multiple alternatives of models, hyperparameters,
        data transformation, and so on.

\leavevmode\vspace{-5pt}

\subsection{Current Pain Points}\label{sec:interview-painPoint}
The interviewees (F1--F6) shared various inconveniences they had been experiencing.
See Table~\ref{tab:painpoints} for a summary.

(P1) 
Jupyter's web interface may lose its connection
    with the underlying kernel (i.e., Python interpreter)
    due to network issues or system suspensions.
Then, all the (intermediate) data that have been generated
    are erased and cannot be restored.
F5 and F6 typically re-executed cells to recover their work,
    which was tedious.

(P2)
Using Jupyter to compare different alternatives
    often overwrites data that cannot be restored easily.
For example, F4 noted ``If I copy a script from somewhere else into Jupyter and execute it,
it might override some variables. 
It's hard to track. You also need to rerun some cells to refresh the variables before proceeding.''


(P3, P4. P5)
In Jupyter, it is often hard to understand
    which of the past executions caused the current (abnormal) result.
One mitigation mechanism is to print as many variables as possible
    every time a code block is executed, which is cumbersome.
Of course, we cannot print variables for past states. As F1 noted ``Many times the code only works for the current moment because notebook is messy and you don't know how to execute the cells in the right order to reproduce the current variables. You have to mentally maintain the order when developing in notebook.''

(P6) A session cannot be moved between machines,
    making dynamic scaling challenging.
Running Jupyter for large-scale distributed machine learning
    is uncommon in the industry.

\subsection{How Versioning for Both Code
and Data Can Help Them}\label{sec:interview_checkpointing}

\paragraph{Benefits of Versioning}

When we shared the idea of supporting nonlinear data science
    through versioning both code and data,
all interviewees agreed on its potential usefulness.
Specifically,
    many of the current pain points (P1--P4)
        can almost be directly solved by the new capability.
P5 can be partly addressed
    as users can browse their exploration history
        and the variables from the past states.
P6 can be mitigated as checkpointing
    enables moving Jupyter sessions between machines,
        allowing dynamic resource scaling.

\paragraph{Further Suggestions}

In addition to excitement,
    the interviewees shared suggestions for further improvements.
We summarize this feedback in Table~\ref{tab:suggestions},
    and highlight a few novel ones as follows.
S1--S3 suggested independent versioning of code and data.
\fix{In particular, S2 would like to compare changes in data sources (e.g., file names, dataset versions, data fetching parameters) and/or changes in cell code.}
F1 mentioned an interesting use case in which one needs to keep the code 
but revert variables: 
``There's a mistake in my code, and I need to revert my variable state from the mistake, and then correct and rerun the code after it.'' 
F2 brought up another case:
    only the input data changes while keeping the code the same.

Some expressed concerns that a commit history may become too complex
    as exploration continues.
This may require additional tools for easier exploration of histories.
Also, they mentioned a need for tracking changes in a specific variable.
Moreover, comparing (i.e., diff-ing) two commits can be useful
    for easier understanding.



\section{System Overview: Use Cases, User Interface, and Architecture}
\label{sec:design}

In this section, we describe how we design a new user interface (UI), called \system,
    to address the suggestions from Table~\ref{tab:suggestions} (i.e., S1-S6).
Specifically, we explore specific use cases (\S\ref{sec:design-usecases}),
    and derive a UI for supporting those use cases (\S\ref{sec:design-ui}).
Finally, we describe how \system interacts with the underlying Jupyter system (\S\ref{sec:design-arch}).


\begin{figure*}[t]
    
\includegraphics[width=0.95\textwidth]{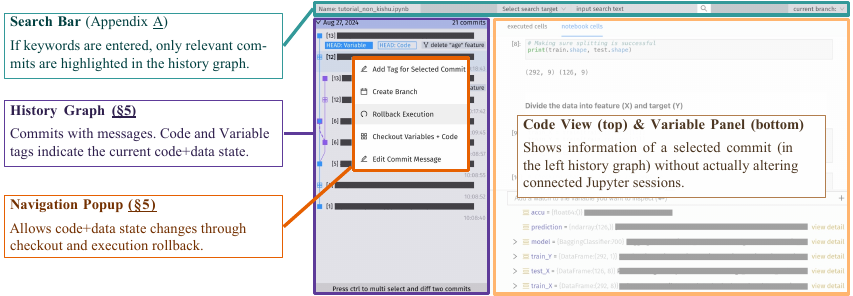}
    
\vspace{-2mm}
\caption{
\system user interface.
The history graph (purple box) shows past commits
    with code+data tags,
    allowing users to quickly grasp
        the current state.
The code and variable panes (yellow box) display
    the information of a selected commit (in the history graph).
From any past commit,
    users can load data only (i.e., execution rollback)
        or load both code and data (i.e., checkout)
    using the navigation popup (red box).
Then, the Code and Variable tags move appropriately.
}
\Description{This figure shows Kishuboard User Interface
The interface includes:
Search Bar: Filters commits based on keywords.
History Graph: Displays past commits with code and data states.
Navigation Popup: Allows switching between code+data states (rollback or checkout).
Code View & Variable Panel: Shows details of the selected commit without changing the active Jupyter session.
This interface helps users manage and explore different versions of computational notebooks effectively.
}
\label{fig:ui:kishuboard}
\end{figure*}


\subsection{Potential Use Cases}
\label{sec:design-usecases}

\paragraph{Manage Alternatives} 
Alice completes data transformation, feature engineering, 
    and model selection.
Later,
    a new feature engineering method occurs to her.
Normally, trying it would require overwriting her data, forcing her to reload the entire dataset and re-execute cells.
With \system, Alice can create a new branch to test the alternative method without losing her current progress. 

\paragraph{Understand What Happens and Reverse Debug} 
Bob shares his work with colleagues. 
Without a tool like \system, understanding the order in which cells were executed can be challenging, especially since cell execution can happen out of order~\cite{wang2020assessing,head2019managing,kery2019towards,kery2018interactions,gitbook}. 
With \system, Bob and his colleague can view a detailed execution log
    and understand the history.
Unlike existing extensions~\cite{head2019managing,kery2019towards},
    \system allows it without
        repeatedly re-executing cells and printing outputs~\cite{chattopadhyay2020s}.

\paragraph{Quick Work Recovery When Restart}

Computational notebooks run on top of the underlying kernels (e.g., Python, R).
If the kernel crashes~\cite{de2024bug,chattopadhyay2020s,li2023elasticnotebook},
    all the cells need to be re-executed from the beginning
        to restore states, which is time-consuming and error-prone~\cite{wang2020assessing,reimann2023alternative,rule2019ten,singer2020notes}.
\system allows for quickly restoring the work by checking out the latest commit.

\paragraph{Isolated Exploration Branch} \system can also improve code quality. 
Exploring a new path often involves creating/renaming variables,
    accidentally overwriting important variables.
\system effectively isolates the code for different branches.
Users only need to check out a commit to restore both code and data states. 

\subsection{User Interface}
\label{sec:design-ui}

The design goal of \system is to let
    users easily manipulate the current state of computational notebooks
    through two independent knobs (i.e., code and data)
while allowing them to easily understand 
    how the notebook states have been changed over time.
We aim to achieve this high-level goal
    by introducing minimal changes to commonly used UIs
        for intuitive understanding and easy adoption.

\paragraph{Components}

\system UI comprises four main components as shown in Fig~\ref{fig:ui:kishuboard}: 
the code view (right), the history graph (left), the search bar (top), and the variable panel (bottom). 
The \textbf{code view} presents the notebook's state at a specific point in time, displaying the content of various cell types (e.g., code cells, markdown cells) along with their outputs. 
Users can toggle the code view to an alternative mode 
    that shows the sequence of executed cells.

The information in the code view corresponds to a commit chosen by the user, 
while the \textbf{history graph} shows the relationships among commits. 
Users can navigate through the history graph to browse commits
    and select an individual commit for details. Based on our technical descriptions presented in \S\ref{sec:version}, our UI visualize the two-dimensional data+code version history in one-dimensional space. This has the following benefits:
\begin{enumerate}
    \item \textsf{Managing Large Number of Commits.} \quad
By linearizing code+ data versions into a one-dimensional axis, \system effectively visualizes numerous commits within the commit history graph.
To further enhance scalability, \system employs auto-folding, which groups and folds related commits to reduce vertical space usage, particularly for long branches.~(Appendix~\ref{sec:others_fold})
Given that each commit occupies a row in the commit history graph, \system implements infinite scrolling, where commits are lazily loaded and prefetched as users navigate through the history.
This approach ensures efficient browsing without overwhelming the UI with excessive data.
\item \textsf{Managing Large Number of Branches.} \quad To handle many branches, \system can leverage practices from existing GUIs for version control systems~\cite{github,fork} to display recently active branches and collapse other branches. 
\system adopts an existing method in gitamine~\cite{commitgraph} to draw the commit history graph, which prevents crossings across different branches.
\system can also automatically select top-K most relevant branches based on their attributes and proximity to the current state. 
Additionally, \system includes a branch management panel that lists all branches, helping users quickly locate the commit where the branch head resides.
\end{enumerate}

The \textbf{search bar}
    offers searches based on commit content, commit message, and names of the variables modified by the commit.
Commits that match with the search are highlighted in the history graph. One specific use case for search is to find all commits that changed a specific variable: Users can track how the variable changes by inspecting those commits highlighted as search results.

Finally, the \textbf{variable panel} shows the state of variables at a selected commit. 
This panel provides details such as the variable name, type, and a string representation of its value, 
along with extended representations for special types, 
such as table formatting for data frames. 
Users can also expand variables 
    to inspect nested attributes. To ensure scalability, \system presents variable names and associated statistics in a concise table format.
For workflows involving many variables, the table is paginated, showing only a subset at a time.
To aid navigation, a search bar allows users to find and pin specific variables, which are displayed at the top of the table.
For large variables, \system truncates their displayed content, maintaining a balance between usability and scalability.

\paragraph{Two Knobs: Code and Data}

Unlike other systems (e.g., git),
    Kishu-board lets users navigate in the two-dimensional code+data space.
Specifically, the users can load data (and variables) from a past commit
    while retaining the current code,
        effectively achieving execution rollbacks.
Using \system's UI (Fig~\ref{fig:ui:kishuboard}),
    users can perform this execution rollback
        by right-clicking a target commit and selecting ``Rollback execution''.
The actual Jupyter state changes immediately
    as \system lets it load variables from the target commit.
Jupyter execution counters---displayed by each cell to
    indicate orders---update appropriately.
Likewise, the users can load both code and data states from
    a past commit by
        selecting ``Checkout Variables + Code'',
    which effectively time-travels the entire system state
        (i.e., both code and data).


\begin{figure}[t]
    \centering
    \includegraphics[width=0.4\textwidth]{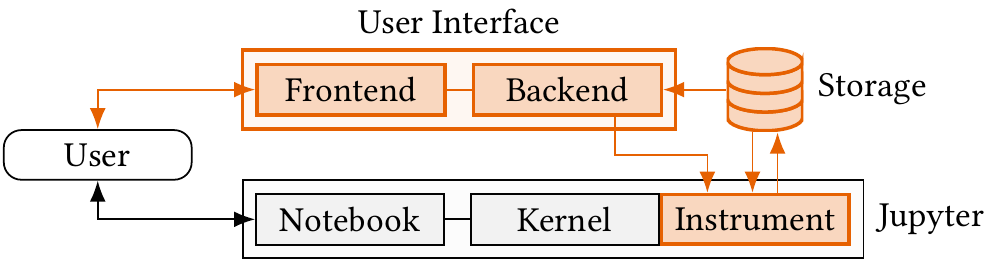}
    \Description{
The diagram illustrates Kishuboard's architecture, highlighting three main components:
1. User. Interact with the User Interface component to manipulate the Kishuboard, and interact with the notebook component to edit and execute notebook cells.
2. User Interface. Composed of a frontend and a backend, with the user component interacting with the frontend and the storage component interacting with the backend.
Storage: Interact with the backend.
Instrument: Works within the Jupyter environment, connecting the kernel and backend to facilitate interaction and data management.
This architecture enables seamless integration between the user interface, storage, and Jupyter's computational environment.
    }
    \caption{\system's system architecture consisting of user interface, storage, and instrument \textcolor{cAmain}{highlighted}.}
    \label{fig:architecture}
\end{figure}


\subsection{System Architecture}
\label{sec:design-arch}

To support the UI above, our system is designed to have three main functions:
(1) \textbf{Collecting Information}: When \system is enabled on a notebook, it attaches an \emph{instrument} to the Jupyter kernel. This instrument detects cell executions and captures all relevant notebook state information, including cell content, cell outputs, and variable values. It then persists this data in \system's \emph{storage}, which also maintains additional metadata, such as commit histories and branch positions.
(2) \textbf{Browsing Code and Data}: The \emph{UI} of \system reads from the storage and preprocesses this information in its \emph{backend} server, before displaying it on the \emph{frontend} web interface. When users interact with various UI components---such as selecting a commit---the UI requests the additional data and updates the displayed state accordingly.
(3) \textbf{Altering Notebook State}: Certain interactions, such as checking out a commit or rolling back executions, require changes to the Jupyter kernel and notebook. In these cases, \system's UI sends commands to its instrument, which reads the necessary information from storage and modifies the kernel and notebook states as needed.

\paragraph{Implementation} \system's UI is a web application consisting of a frontend written in TypeScript with React and a backend written in Python with Flask. \system's storage is an embedded database library written in Python. We integrate \system's instrument onto Jupyter based on its extension framework~\cite{ipython_callbacks}.



\section{Formalizing Code+Data Versioning: Model, Consistency, and Visualization}
\label{sec:version}

Users need the ability to recover code and data separately (Tabel~\ref{tab:suggestions}), but certain recovery actions can confuse the users and lower productivity. To address this, we first formally define the data model (\S\ref{sec:version-set}) and its state transitions (\S\ref{sec:version-navigate}) to identify cases where checkouts can be problematic (\S\ref{sec:version-problem}). To overcome these issues, we introduce a consistency property to prevent problematic actions (\S\ref{sec:version-property}).

\subsection{Data Model}
\label{sec:version-set}

\system helps users manage ``code+data'' versions, which broadly encompass notebook and variable states. This section formalizes what \system regards as code, data, versions, and commits.

\paragraph{Code and Data} On the high level, the \textbf{code state} consists of notebook cells (e.g., code and Markdown cells) and outputs; 
    the \textbf{data state} consists of a variable map (i.e., a mapping from variable names to their values) as well as execution history.
Specifically,
let $C_i$ be the $i$-th code state and $D_j = (X_j, H_j)$ be the $j$-th data state consisting of variable map $X_j$ and execution history $H_j$. An execution history is a list of strings $H_j = [h_j[1], h_j[2], \dots]$ where $h_j[n]$ denotes the $n$-th executed cell.

\paragraph{Versions} A version is a pair of code and data state; that is, $V_k = (C_i, D_j)$. 
Within \system,
code and data states are managed independently. 
When users \emph{edit} any of the notebook code or Markdown cell, the code state progresses while the data state remains the same, from $V_k = (C_i, D_j)$ to $V_{k'} = (C_{i'}, D_j)$. 
When users \emph{execute} a code cell $c$, they may update both code and data states by generating a new cell output and mutating the variable map from $V_k = (C_i, D_j)$ to $V_{k'} = (C_{i'}, D_{j'})$. In this latter case, the execution history will include the newly executed code cells, $H_{j'} = H_j \oplus [c]$, where $\oplus$ is the concatenation between two lists.

\paragraph{Commits} A commit is a version that \system persists.
By default, \system creates and persists a version into a commit after every cell execution. In addition, \system also allows users to create a commit manually.
In general, commits are a subset of all possible versions. 
We denote $\mathcal{V}$ as the set of all versions and $\mathcal{U}$ as the set of commits; $V_i \in \mathcal{U}$ implies that \system persisted $V_i$.

\subsection{Navigating in Code+Data Versions}
\label{sec:version-navigate}

\begin{figure*}[t]
    \centering
\begin{subfigure}[t]{\textwidth}
\centering
\includegraphics[width=1.0\textwidth]{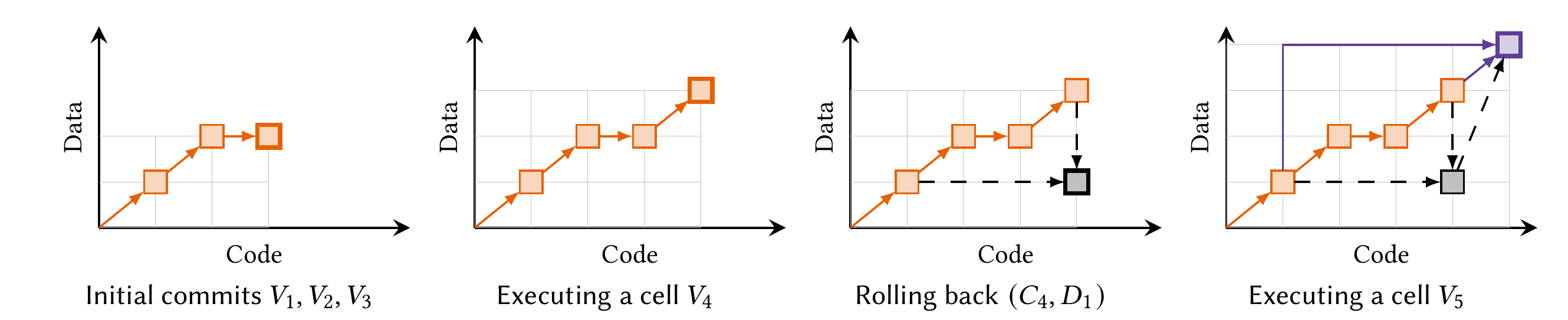}

    \Description{
    The figure shows the process of navigating through different versions in a 2-dimensional axis system.
    }
    \caption{\fix{Navigation in two-dimensional axis system: An example of 4 events from left to right. Thick rectangles denote head versions.}}
    \label{fig:nav2d}
\end{subfigure}
\begin{subfigure}{\textwidth}
    \centering
    \includegraphics[width=1.0\textwidth]{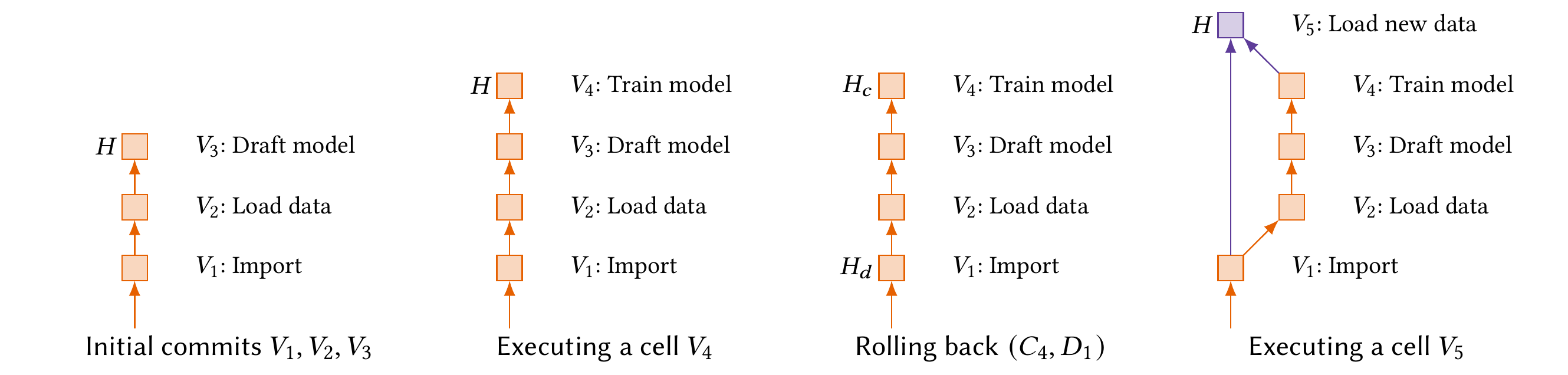}

    \Description{
    The figure shows the process of navigating through different versions in a 1D axis system.
    }
    \caption{\fix{Navigation in one-dimensional axis system: $H$ denote the current head's code+data version, while $H_c$ and $H_d$ denote code and data versions.}}
    \label{fig:nav1d}
\end{subfigure}
    \vspace{-2mm}
    \caption{\fix{Examples of navigation in two-dimensional and one-dimensional axis systems}}
    \label{fig:combined-navigation}
\end{figure*}

Besides editing the notebook and executing a code cell, notebook users may navigate to different versions through \system by checking out a commit and rolling back executions.
\system visualizes these changes using a one-dimensional axis system.

\paragraph{\fix{Two-dimensional Versions} and History} By definition, versions are two-dimensional; a version $V_k$ is described by a coordinate of $C_i$ and $D_j$. For example, in Fig~\ref{fig:combined-navigation}a, Bob, a data scientist, has created three committed versions $V_1 = (C_1, D_1)$, $V_2 = (C_2, D_2)$, and $V_3 = (C_3, D_2)$ from his notebook. When he executes a cell, he creates another committed version $V_4 = (C_4, D_3)$.
In theory, the notebook could be in any combination of code and data states such as $(C_1, D_2)$ or $(C_3, D_3)$, denoted as a grid in Fig~\ref{fig:combined-navigation}a:
Bob can roll back execution from $V_4 = (C_4, D_3)$ to $(C_4, D_1)$.

Each committed version $V_k$ has a code parent $P_c(V_k)$ and a data parent $P_d(V_k)$ which indicates the previous respective state. Both parents may belong to the same or different committed versions. For example, $V_3$ is both code and data parent of $V_4$ because $V_3 = (P_c(V_4), P_d(V_4))$. On the other hand, when Bob reloads his data after rolling back execution, he would create a commit $V_5$ with a code parent $P_c(V_5) = C_4$ and a data parent $P_d(V_5) = D_1$. The collection of versions and their parents form a history $(\mathcal{U}, P_c, P_d)$, embedded in the code+data space.

\paragraph{Visualization Commit History Graph \fix{in One-dimensional Axis System}} Instead of displaying commits using the \fix{two-dimensional} code+data space, \system opts to visualize the history using a \fix{one-dimensional} \textbf{commit history graph} that represents commits using nodes and their parent relationship using edges.
Fig~\ref{fig:combined-navigation}b shows the commit history graph equivalent to the \fix{two-dimensional} history in Fig~\ref{fig:combined-navigation}a.
By doing so, we simplify the user's learning curve due to its resemblance to familiar Git interfaces~\cite{sublime_merge,git_gui_clients,gitkraken_best_git_gui}. 
\fix{Furthermore, this approach utilizes space more efficiently: one-dimensional axis system takes $O(NB)$ to display $N$ commits and $B$ branches (typically, $B \ll N$) while the two-dimensional system takes $O(N^2)$.}

In addition, this commit history graph denotes the head commit $H = (H_c, H_d)$ (i.e., the current version) by labeling the corresponding commit(s). If the current head's code and data states belong to different commits, the commit history graph labels the two commits accordingly (e.g., the third step in Fig~\ref{fig:combined-navigation}b). As a result, the commit history graph displays a commit with two different parents (e.g., $V_5$ in Fig~\ref{fig:combined-navigation}b) by connecting it directly to its parents.

\subsection{Problematic Checkouts}
\label{sec:version-problem}

Theoretically, \system could allow checking out any pair of code state $C_i$ and data state $D_i$; however, we observe problematic checkouts that create inconsistent states,
    as illustrated in Fig~\ref{fig:problematic_checkouts}. 

\begin{figure}[t]
    \begin{subfigure}[b]{0.45\textwidth}
        \centering 
        \includegraphics[width=\linewidth]{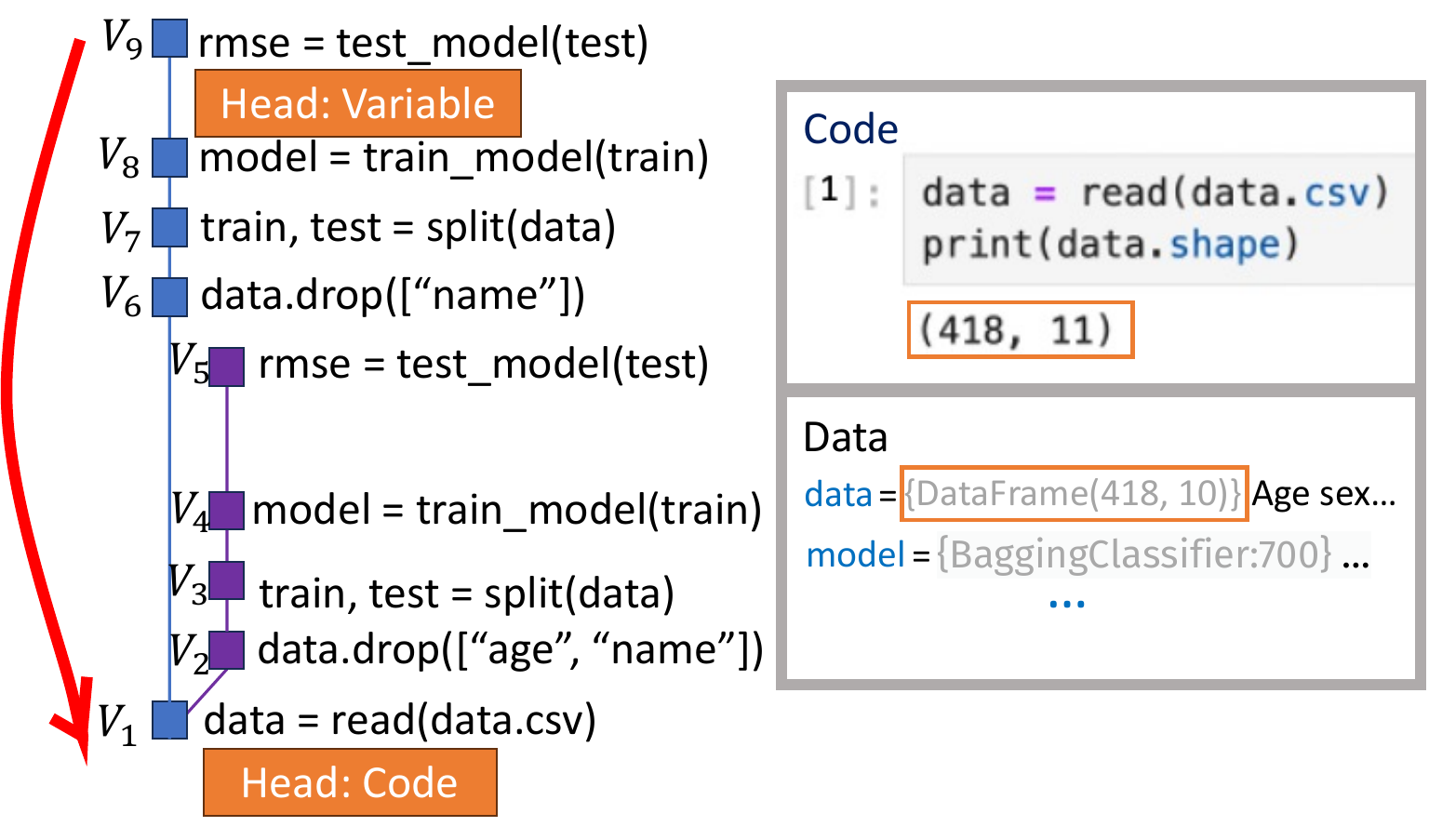}
        \vspace{-4mm}
        \caption{Checkout Only Code.}
        \label{fig:checkout_only_code}
    \end{subfigure}
     \hfill
    \begin{subfigure}[b]{0.45\textwidth}
        \centering 
        \includegraphics[width=\linewidth]{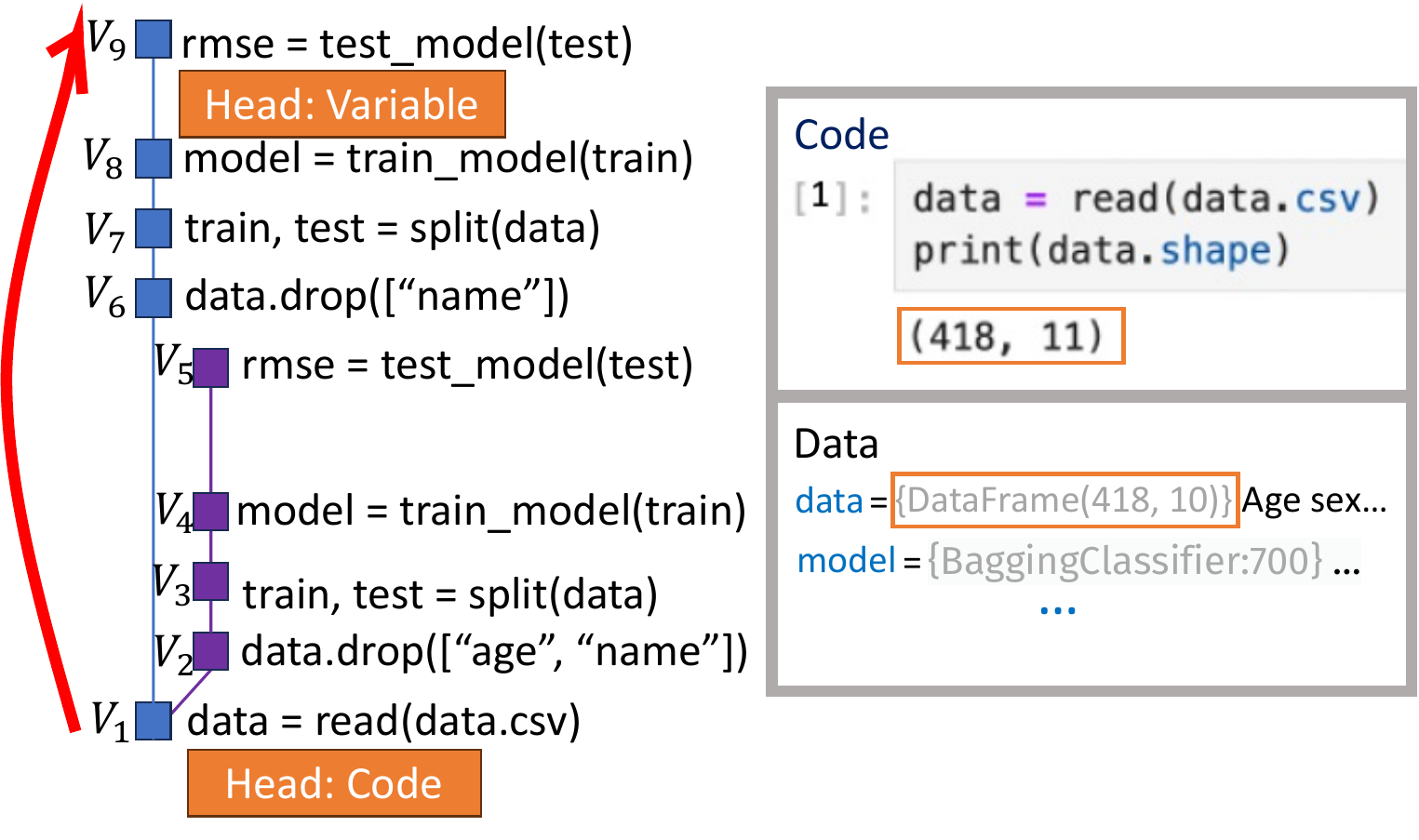}
        \vspace{-4mm}
        \caption{Checkout Future Data.}
        \label{fig:checkout_future_data}
    \end{subfigure}

    \vspace{2mm}
    \begin{subfigure}[b]{0.45\textwidth}
        \centering 
        \includegraphics[width=\linewidth]{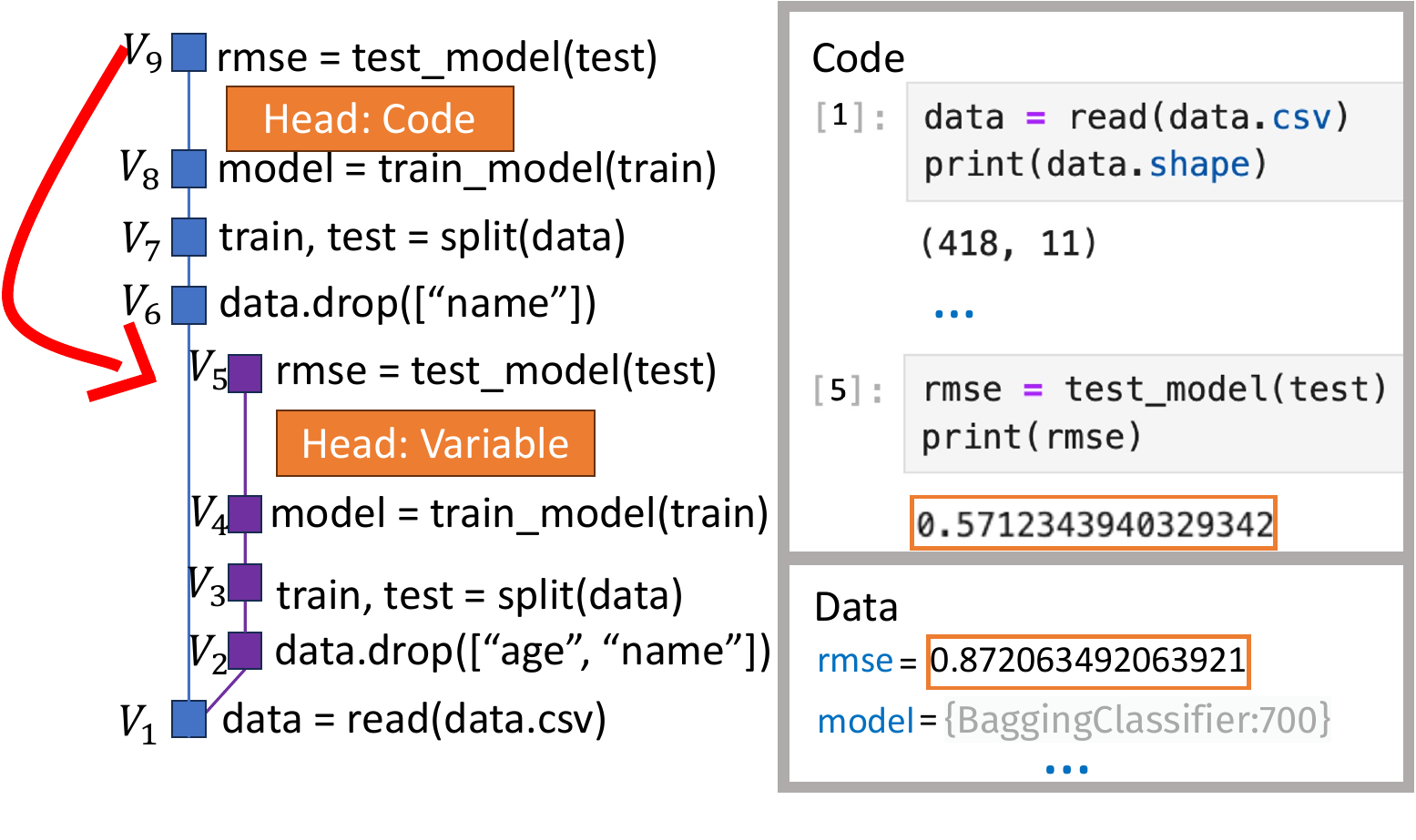}
        \vspace{-4mm}
        \caption{Checkout Unrelated Data.}
        \label{fig:rollback_unrelated_data}
    \end{subfigure}

    \vspace{-2mm}
    \Description{The figure shows three scenarios where checking out different versions disrupts code and data consistency, with orange rectangles highlighting inconsistencies between the code and data states.}
    \caption{\fix{Three problematic checkouts that break consistency. The inconsistency of code and data is highlighted in orange rectangles. \textbf{(a) Checkout Only Code}: the code(including cell output) is checked back to reflect data in $V_1$, but the data is still in $V_9$. \textbf{(b) Checkout future data}: the data is checked out to $V_9$, but the code still reflects data in $V_1$. \textbf{(c) Checkout Unrelated Data}: the data is checked out to $V_5$ on another branch, but the code still reflects the current branch's data.}}
    \captionsetup{skip=-10pt}
\label{fig:problematic_checkouts}
\end{figure}

\paragraph{Problems When Checking Out Only Code} In Fig~\ref{fig:problematic_checkouts}a, consider a scenario where only the code is checked out to $V_1$ while the data is still in $V_9$. Based on the code, users misunderstand that \texttt{df} holds the data from the original \texttt{data.csv}. However, since the kernel data is still $V_9$, the \texttt{name} feature is dropped from the original \texttt{df}.

\paragraph{Problems When Checking Out Future Data} Similarly, checking out data from $V_1$ to $V_9$ but keeping the code in $V_1$ will lead to the same inconsistency as checking out only code.

\paragraph{Problems When Checking Out Unrelated Data} Alice is at $V_9$, and she checks out only the data to $V_5$ from a different branch in history. This is also confusing because the code is still at $V_9$, and from the code, users assume the \texttt{age} feature is never dropped, and the model is trained with the \texttt{age} feature, with the \texttt{rmse} of \texttt{0.5123}; however, this is not the case for our variables in kernel, which is already in $V_5$. The nature of the problem is that $V_5$ and $V_9$ are on different branches, so their histories are unrelated.

\subsection{Consistency and Safe Checkouts}
\label{sec:version-property}

$V_k = (C_i, D_j)$ is \textbf{consistent} if and only if the cell outputs in $C_i$ agree with the results of execution according to the execution history $H_j$. 
Formally, for all executed cell $c \in C_i \cap H_j$, let $n$ be the order the cell $c$ is executed (i.e., $h_j[n] = c$), $V_k = (C_i, D_j)$ is consistent if and only if $\text{Output}(c) = \text{Exec}([h_j[1], h_j[2], \dots, h_j[n]])$. Here $\text{Exec}(H)$ is the result of execution over a list of code cells $H$ and $\text{Output}(c)$ is the output of the cell $c$.

\begin{theorem} \label{theorem:commit-consistent}
    All commits $V_k \in \mathcal{U}$ are consistent.
\end{theorem}
\begin{proof}
    We prove the statement by induction. Recall that the empty notebook version $V_0 = (C_0, D_0)$ where $C_0$ contains no cells and $D_0$ contains no variables is consistent because the set of executed cells is empty $C_0 \cup H_0 = \varnothing$.
    
    All commits are subsequent versions of the empty notebook version by editing or executing notebook cells; therefore, it suffices to prove that editing or executing notebook cells from a consistent version $V_k = (C_i, D_j)$ makes a consistent version $V_{k'} = (C_{i'}, D_{j'})$. If a cell $c \in C_i$ is edited, $c$ becomes a non-executed cell in $V_{k'}$; $V_{k'}$ remains consistent because of $V_k$ consistency on other executed cells. If a cell $c \in C_i$ is executed, it becomes an executed cell $c = h_{j'}[n] \in H_{j'}$ and its output is updated with the execution result $\text{Output}(c) = \text{Exec}([h_j[1], h_j[2], \dots, h_j[n]])$, which satisfies the consistency condition. Hence, editing or executing notebook cells preserves consistency.
\end{proof}

Problematic checkouts (\S\ref{sec:version-problem})
are confusing because they can induce inconsistent code+data states.
Checking out only code may introduce inconsistent cell outputs from a different sequence of cell executions. Checking out future data may contain missing cell outputs. Lastly, checking out unrelated data on a different branch may admit execution results that conflict with cell outputs.
To prevent potential confusion, \system implements a safeguard that disallows these unsafe checkouts.
As a result, \system allows two types of safe checkouts:


\paragraph{Checking Out Both Code and Data} Users can safely check out both code and data states from any commit. Thanks to Theorem~\ref{theorem:commit-consistent}, checking out both code and data from a commit always produces a consistent version. For example, in Fig~\ref{fig:consistent_checkouts}a, when we are on $V_9$ and want to continue exploring the case where \texttt{age} feature is dropped, we can check out both code and data to $V_5$ and continue from there.

\paragraph{Checking Out Past Data} Users can safely check out a data state from a past commit (i.e., rollback execution). For example, in Fig~\ref{fig:consistent_checkouts}b, at $V_9 = (C_9, D_9)$, the data scientist checks out the data state to $D_1$ and keeps the code state. Checking out only data generates a new version $(C_9, D_1)$. With $C_9$ kept there, the user does not need to rewrite cells that are shared between different explorations (e.g., data splitting, model training, and model evaluating cells in Fig~\ref{fig:consistent_checkouts}b), while being able to modify other cells.

\begin{theorem}
    If $V_k = (C_i, D_j)$ is consistent, the version after checking out a past data $V_{k'} = (C_i, D_{j'})$ is consistent.
\end{theorem}

\begin{proof}
    A past data state $D_{j'}$ contains an execution history that is a prefix of the current one, $H_{j'} \subseteq H_j$. 
    For all executed cells in the checked-out version $c \in C_i \cap H_{j'} \subseteq C_i \cap H_j$, its output matches with the execution result:
    \begin{align*}
        \text{Exec}&([h_{j'}[1], h_{j'}[2], \dots, h_{j'}[n]]) \\ 
        &= \text{Exec}([h_j[1], h_j[2], \dots, h_j[n]]) \\
        &= \text{Output}(c) \qedhere
    \end{align*}
\end{proof}

\fix{Note that checking out past data is safe even when executing modified cells. For example, when one modifies an existing cell $c$ into $c'$ to create $V_{10} = (C_{10}, D_9)$ so that $c \notin C_{10}$ and $c' \in C_{10}$, executes $c'$ to create $V_{11} = (C_{11}, D_{11})$, and checks out past data to $D_9$, the resulting version $(C_{11}, D_9)$ remains consistent regardless of the modified cell because the modified cell is no longer an executed cell $c' \notin H_9$ and the original cell is no longer an existing cell $c \notin C_{11}$.}

\begin{figure*}[t]
    \centering
    \begin{subfigure}[b]{0.45\textwidth}
        \centering 
        \includegraphics[width=\linewidth]{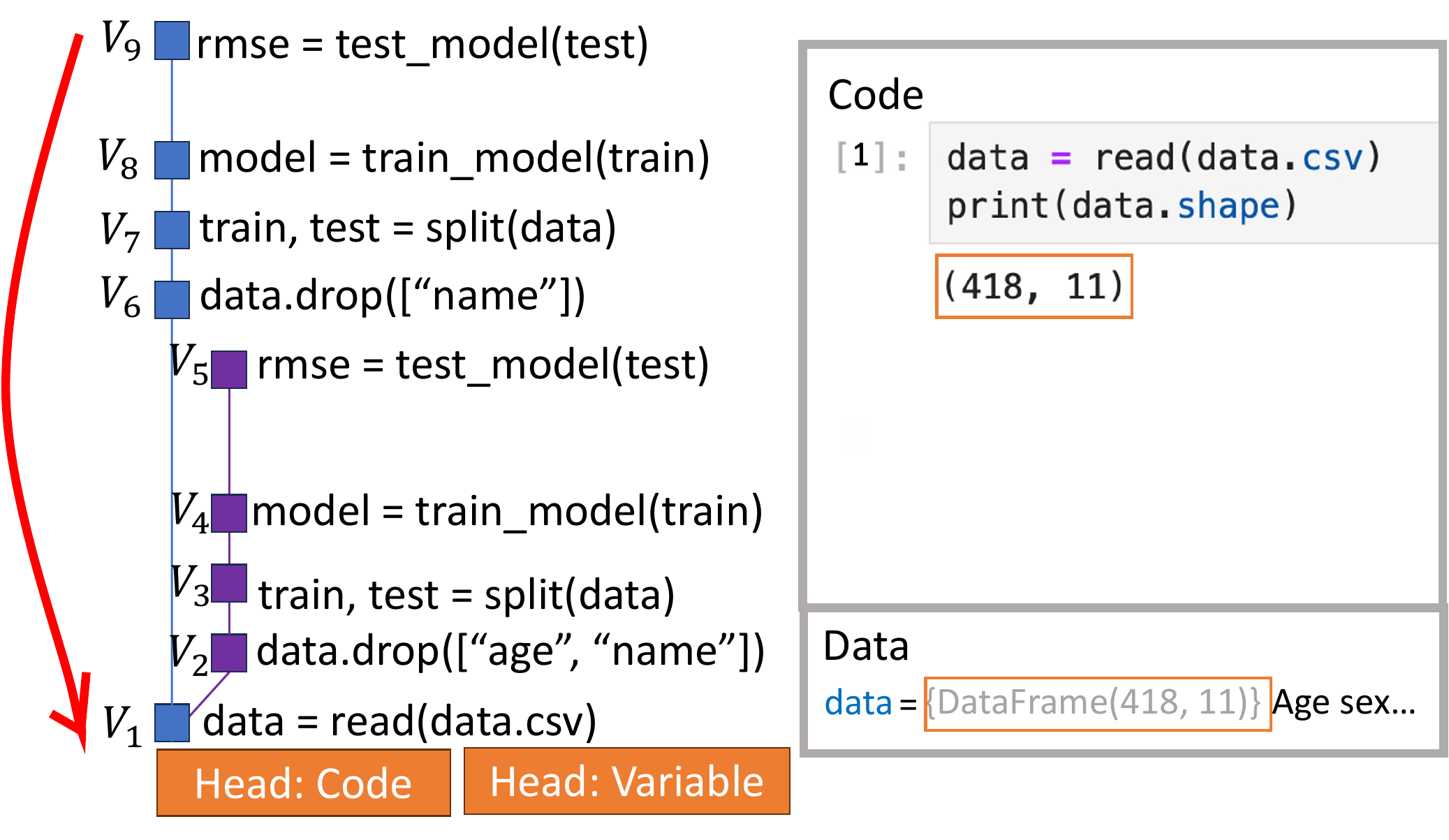}
        \vspace{-4mm}
        \caption{Checkout both code and data.}
        \label{fig:consistent_checkout_both}
    \end{subfigure}
    \qquad
    \begin{subfigure}[b]{0.45\textwidth}
        \centering 
        \includegraphics[width=\linewidth]{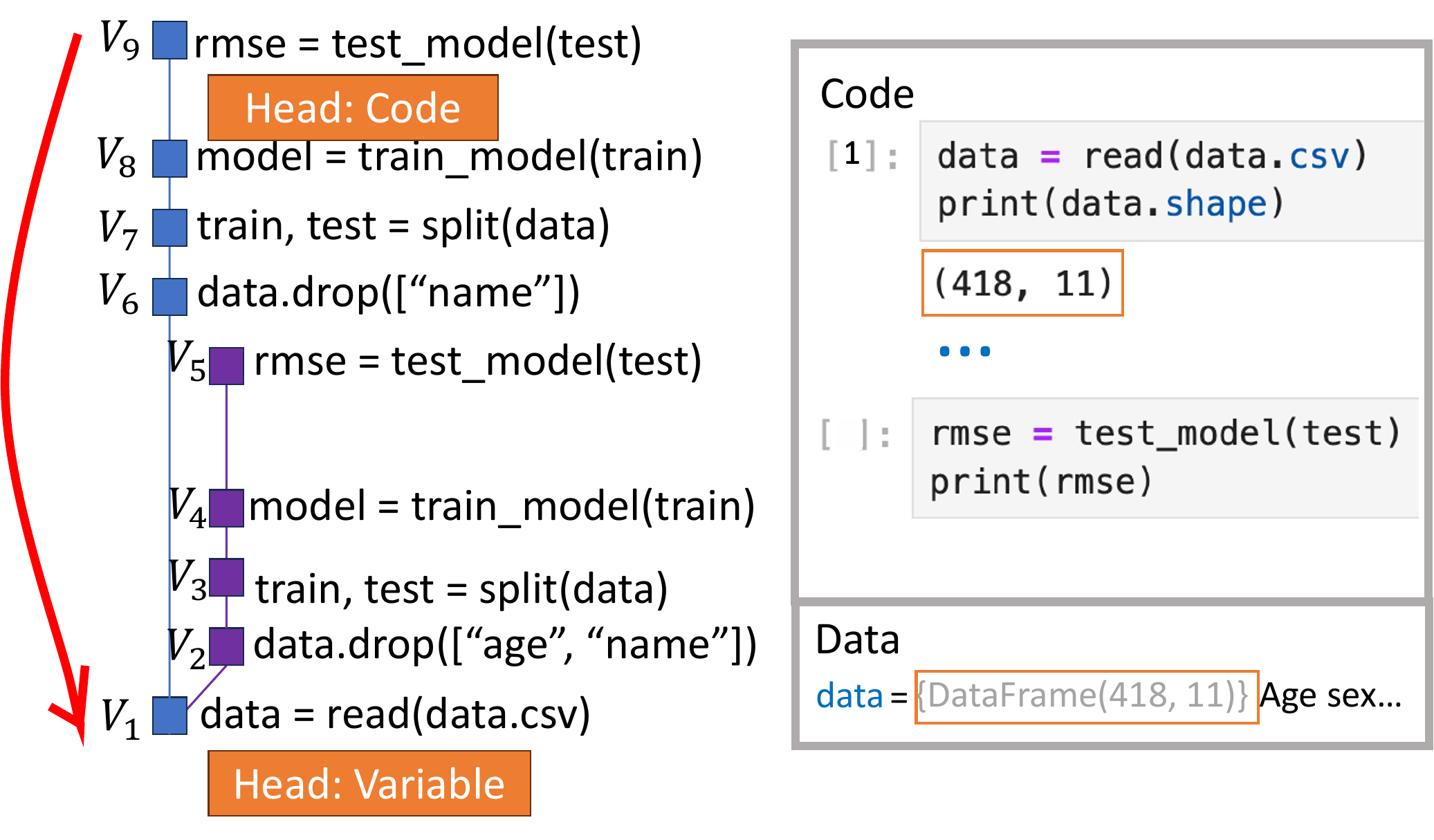}
        \vspace{-4mm}
        \caption{Checkout past data.}
        \label{fig:consistent_rollback_var}
    \end{subfigure}
    \hfill

    \vspace{-2mm}
    \Description{The figure illustrates two scenarios of consistent checkouts: Checkout both code and data and Checkout past data.}
    \caption{\fix{Examples of two types of consistent and safe checkouts. Note that checking out past data rewinds execution number(s).}}
\label{fig:consistent_checkouts}
\end{figure*}



\section{User Study Design}
\label{sec:evalmethod}

We aim to answer two core questions: 
    (1) Does \system improve productivity in data science tasks? 
    (2) Are the \system's UI designs intuitive and easy to interact with?
We study these with a user study with 20 students. 
These students are randomly assigned to the experimental group with \system and 
    the control group without \system 
    (\S\ref{sec:evalmethod-participants}). 
For both groups,
we provide an identical testing environment with 
    the same tutorials, data science tasks, reference code, and instructions (\S\ref{sec:evalmethod-materials}). 
We observe how the users perform the task, measure the time they spend on each part, 
    and collect feedback through exit surveys (\S\ref{sec:evalmethod-procedure}). 
Lastly, we discuss our efforts to minimize the gap between
    this lab experiment and real-world data science tasks (\S\ref{sec:evalmethod-challenges}). The University Ethics Review Board approved human subject experiments and the formative study.

\subsection{Participants}
\label{sec:evalmethod-participants}

We recruited 20 students from a course covering the intersection of machine learning and data management: 1 undergraduate student and 19 graduate students. 
The same amount of extra credits was given 
    as an incentive for participating in this study,
    regardless of their performance.
With ``Artificial Intelligence'' as a course prerequisite, 
all participants had backgrounds in data science workflows 
    such as  extract-transform-load, 
    feature engineering, model training, etc. 
Also, all the participants already had completed 
    course assignments related to Jupyter and Python.

The participants were randomly assigned into 4 (sub-)groups as in~Table~\ref{tab:groups}. 
All four groups performed the same tasks,
    but differed only in the use of \system and workload scales.
The only difference between lightweight and compute-intensive tasks
    was artificially injected delays for certain cell executions 
        (e.g., data loading, training models);
    however, the contents of the datasets and trained models were exactly the same.


\begin{table}[t]
\caption{Participant group assignments. 
    20 participants were randomly
        assigned to one of four subgroups.
    } 
    \label{tab:groups}

\vspace{-2mm}

\centering
\small
\begin{tabular}{lll}
    \toprule
    \textbf{Group Name} & 
    \textbf{\CellWithForcedBreak{\system \\Used?}}  & 
    \textbf{\CellWithForcedBreak{Workload Scale (Running the \\entire notebook end-to-end takes)}} \\
    \midrule
    KL & Yes & Compute-intensive \quad (332 secs) \\
    CL & No (control) & Compute-intensive \quad (332 secs) \\
    KS & Yes & Lightweight \quad (80 secs) \\
    CS & No (control) & Lightweight \quad (80 secs) \\
    \bottomrule
\end{tabular}
\end{table}



\begin{table*}[ht]
\centering
\footnotesize
\caption{Our strategies to close the gap
    between real-world data science and our user studies.}
\label{table:constraints}
\vspace{-2mm}

\begin{tabular}{p{18em} p{15em} p{25em}}
\toprule
\textbf{Real-world data scientists} & 
\textbf{Discrepancy from Real-world} & 
\textbf{Our user study design to close the gap}                                  \\ \midrule
Write notebook code by themselves & 
    Code is provided by us & 
    Multiple code snippets to choose from~(\S\ref{sec:evalmethod-materials})
    \\[0.5em]
May have backgrounds with their tasks & 
Limited time to understand our tasks &  
Tutorial notebooks in the same structure with easier tasks~(\S\ref{sec:evalmethod-materials})
    \\[0.5em]
More familiar with Jupyter and other tools & 
No experience with \system & 
Tutorials have step-by-step instructions for \system~(\S\ref{sec:evalmethod-materials})
    \\[0.5em]
Test various models and datasets
    & Only a few scenarios can be tested
    & ``compute-intensive'' and ``lightweight'' tasks for diversity
    (\S\ref{sec:evalmethod-procedure})
    \\[0.5em]
    \bottomrule
\end{tabular}
\end{table*}


\subsection{Environments, Training, and Task Design}
\label{sec:evalmethod-materials}

\paragraph{\textbf{Environment}} We pre-installed both Jupyter and \system 
    on the testing laptop. 
The laptop was a MacBook Pro with 16 GB memory and a 2.3 GHz Intel Core I7 processor. 
We extended the laptop with an LG 34-inch 21:9 UltraWide 1080p monitor for the user study. 
We 1:1 split the screen into two parts, 
    one for Jupyter and the other for \system;
    participants could see both simultaneously.

\paragraph{\textbf{Introduction Video}} 
We prepared a 4-minute video to introduce \system.
The video gave users a basic idea of \system and how it worked. 
Users in KS and KL watched this video at the beginning; 
those in CS and CL watched it after the formal tasks.

\paragraph{\textbf{Tutorial}} 

Tutorial notebooks were prepared
    to help participants recall Jupyter (for all groups) and 
        try \system before starting 
            actual (i.e., formal) tasks
            (only for KL and KS groups).
The tutorial notebooks had a structure similar to the formal tasks
    for an easier transition,
but their semantics and difficulty were different.
For \system groups (KL and KS), tutorial notebooks had extra modules
    for trying \system beforehand.
This trial session was less than five minutes;
    thus, their exposure to \system was relatively limited
    (compared to Jupyter and Python).

\paragraph{\textbf{Formal Tasks}} 
We prepared two sets of tasks: Workflows I and II.
Each workflow contained multiple (smaller) tasks.

\vspace{0.3cm} 
\noindent\textsf{Workflow I: Build data science models.}  \;
    This mimicked how data scientists would build models to predict Titanic fatalities. 
    Users were instructed to follow a typical data science workflow and choose the methods from \fix{provided code snippets} for data loading, visualization, feature engineering, model training, and evaluation.

    \begin{itemize}
    \item \emph{Task 1: Try alternatives} \;
    We asked to change feature engineering methods and retrain models. 
    This task requires nonlinear exploration because
        some features were dropped previously. 
    \fix{This task corresponds to how F2, F4, F5, and F6 use notebooks to compare different alternatives.}
    
    \item \emph{Task 2: Report old \& new values} \;
    We asked to report old and new model accuracies (RMSE). 
    \fix{This task also corresponds to how F2, F4, F5, and F6 compare alternatives.}
    
    \item \emph{Task 3: Retrieve old variable} \;
    Finally, we asked to export both old and new models for future inspection. \fix{This task corresponds to pain point P2 (Cannot restore overwritten data).}
    \end{itemize}
\vspace{0.3cm} 
\noindent\textsf{Workflow II: Recover work from bugs and system crashes.} \quad 
    While sequentially executing the code blocks in a notebook, 
    participants would find that a bug has corrupted \texttt{data} variable. 

    \begin{itemize}
    \item \emph{Task 4: Identify bug} \;
    We asked to find the root cause of the bug. \fix{This task corresponds to pain point P3 (Hard to debug).}
    
    \item \emph{Task 5: Recover from restart} \;
    After fixing the bug, they were asked to restart the kernel.
    This resembles a system crash.
    Participants were then asked to restore their work. \fix{This task corresponds to pain point P1 (Session suspension deletes data).}
    \end{itemize}
\vspace{0.3cm}

To make the notebooks realistic, all code snippets were taken from Kaggle prize-awarded notebooks, 
with minor modifications of variable names to ensure compatibility.
Users were asked to open and read instructions for one task at a time.
This was to simulate the situation
    where they progressively develop their notebook.

Providing code snippets helps us evaluate the impact of our framework (\system) in a more controlled manner. 
It minimizes the experimental variance coming from participants with different backgrounds (e.g., different levels of familiarity with data science and Python).
We note that although snippets are fixed, participants are free to edit them and explore data in any way combining both conventional copy/paste and \system's checkout and rollback.

\subsection{Procedure}
\label{sec:evalmethod-procedure}

\paragraph{Experiment Groups w/ \system (KS, KL)} 

We first presented a \system introduction video. 
Participants were then asked to study tutorial notebook. 
Later, they worked on the two formal task notebooks. 
Lastly, they filled out an anonymous survey.
\\[-17pt]
\paragraph{Control Groups w/o \system (CS, CL)} 
Instead, participants were first asked to study the tutorial notebook. 
Then, they were instructed to work on the two formal task notebooks. 
Afterward, we presented \system introduction video to participants. 
Lastly, they filled out an anonymous survey.

\vspace{3mm}
\noindent
While the users were working on the formal tasks, 
    we measured the time they spent on each task. 
The user study lasts for about an hour across all participants.

\subsection{For realistic data science environments}
\label{sec:evalmethod-challenges}

While our experiment setup may not exactly replicate real-world data science workflows,
    we employ various mitigation strategies to close the gap.
Table~\ref{table:constraints} describes
    potential differences between real-world data science and our lab environment 
        (the first and the second columns),
    and present our strategies (in the third column). 
For example, real-world data scientists would 
    convert their high-level goals to code by themselves.
In contrast,
    our user study provided participants with
        reference code snippets
        to help them complete tasks in a limited amount of time (i.e., one hour).
To narrow this gap,
    we provided multiple alternative code snippets.
Participants could mimic the process of "writing code" with "their own choice."
Another limitation would be users' lack of experience with \system. 
To mitigate this, we provided training with a tutorial notebook.



\begin{figure*}[ht]
    \begin{subfigure}[b]{0.47\textwidth}
        \includegraphics[]{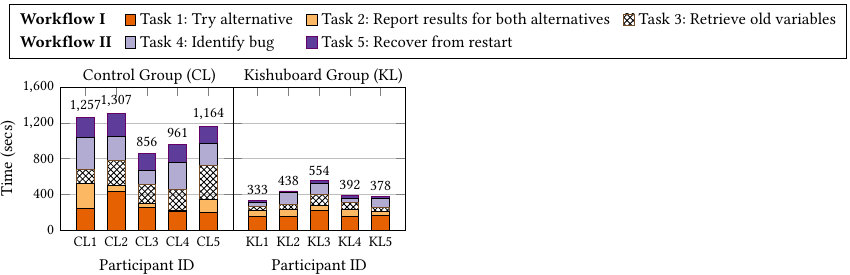}
        \vspace{-2mm}
        \caption{Performance on compute-intensive workloads}
        \label{fig:time_l}
    \end{subfigure}
    \begin{subfigure}[b]{0.39\textwidth}
          \includegraphics[]{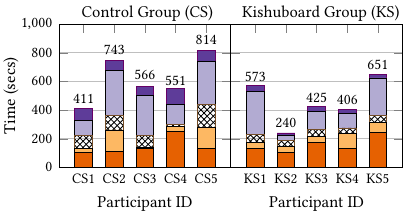}
        \vspace{-2mm}
        \caption{Performance on lightweight workloads}
        \label{fig:time_s}
    \end{subfigure}

    \vspace{-2mm}
    \Description{
    The figure consists of two bar charts comparing the performance of participants in two different groups, "Control Group" and "\system Group," under different workflows for both compute-intensive and lightweight workloads. Each group is assessed based on the time taken (in seconds) to complete various tasks within two workflows.
    }
    \caption{User performance on Tasks 1--5.
    \system helped participants complete tasks more quickly.
    The improvements were more significant for compute-intensive workloads (left).
    For lightweight workloads (right), 
    some people could quickly achieve nonlinear exploration
        by re-running code from scratch,
        reducing the advantage offered by \system.
    }
    \label{fig:time_comparison}
\end{figure*}


\section{User Study Results}
\label{sec:eval}

We conduct a user study to evaluate \system as a nonlinear interactive data science tool. The study shows:
\begin{enumerate}
    \item The proposed code+data space versioning boosts productivity. 
    \system helped users complete tasks more quickly. 
    The improvements were more significant ($62\%$ faster on average) for compute-intensive workloads. For lightweight workloads, the improvement is $25.61\%$ on average.~(\S\ref{sec:eval-productivity})
    \item \system's features are useful for data science workflow. 
    All the features (i.e., checkout code+data, checkout past data, locate commit by exec counter, locate commit by searching, browse exec history and variable info) 
    were considered ``very useful'' or ``useful'' and were used by most participants.~(\S\ref{sec:eval-features})
    \item Our proposed UI design is easy to use but has room for improvement. 
    Most considered it easy to browse commits, locate the ones they needed, 
        trace the variable changes, 
        and understand execution history.~(\S\ref{sec:eval-ui_easiness})
    \item Participants were willing to continuously use or 
        even pay. 
        19/20 participants were "willing" or "very willing" to use. 
        13/20 participants were willing to pay at least \$5 a month.~(\S\ref{sec:eval-willingness})
\end{enumerate}

\subsection{Code+data versioning boosts productivity}
\label{sec:eval-productivity}

We quantify \system's productivity benefits
    by comparing the performance of \system groups (i.e., KL and KS)
        and that of the control groups (i.e., CL and CS) on
        five different tasks (i.e., Tasks 1--5).
The results are shown in Fig~\ref{fig:time_comparison}. 
We observe that \system offered significant productivity benefits, especially 
    when the workload was compute-intensive (Fig~\ref{fig:time_comparison}a).
This was because \system allowed users to explore alternative paths
    without re-executing some of the code blocks.
They could directly check out or roll back executions
    to retrieve any previous data and resume from that point.
This time-saving benefit was lower for lightweight tasks, unexpectedly,
    because there is less time penalty in manually simulating nonlinear exploration.
The time-saving benefit was the biggest for Task 5 (Recover from restart) 
    because re-executing every cell from scratch is time-consuming.
Such a manual restoration strategy can only work for small-scale workflows with a relatively short history
    (like in our user study),
which will become extremely challenging in real-world data science.

The average treatment effect was statistically significant even considering the variance
    in individual performance.
Specifically, the overall performance (i.e., the sum of the times spent on Tasks 1--5) of 
    the KL group was higher than the CL group with a p-value $4.17 \times 10^{-5}$;
    the performance of KS was higher than CS with a p-value of $7.9 \times 10^{-2}$.
We also manually confirmed, based on video recordings of user interactions,
    that those performance variances were caused by natural data exploration, 
        not by any unexpected reasons (e.g., system crashes).
For example, in Task1 (Try alternative), KS4 and KS5 lost their code by checking out both code and variable instead of checking out only data. 
As a result, they needed to re-choose the same code twice. 
We also observed that different CS and CL participants 
    selected different strategies to perform their tasks;
    this had an interesting implication that
        the additional data operations they performed on an earlier task
        often aided their later tasks.
For example, CS4 was relatively slow in Task 1 but was much faster in Task 2 and Task 3. 
It was because this user duplicated every cell and renamed every variable to avoid overwriting, 
    which slowed down trying model alternatives but simplified model comparison and retrieving old models.
Some other participants, like CL2, CL3, CS1, and CS3, did not overwrite the old RMSE 
    when printing the new one by duplicating just the model evaluation cell, 
        helping them complete Task 2 quickly.

\fix{We acknowledge that, in our user study, users may spend relatively smaller amounts of time writing and editing code compared to real-world data exploration. 
If users spend significant time writing code, it could reduce the \system's benefits in end-to-end execution times.
Nevertheless, \system can offer significant productivity benefits in several ways. 
Users can easily retrieve mistakenly overwritten code from checkpoints
    and reuse code from different branches without painstakingly rewriting code 
        and managing multiple exploration paths. }

\subsection{Our features are useful for data science}\label{sec:eval-features}

\begin{figure*}[t]
    \centering
    \begin{subfigure}[b]{1.0\linewidth}
        \centering
    \includegraphics[]{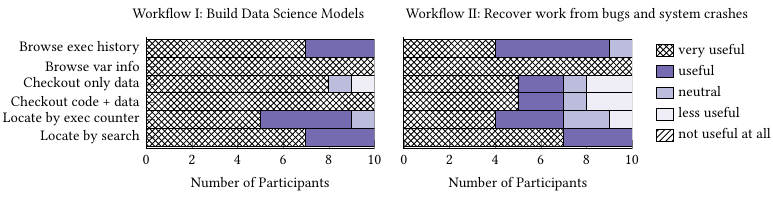}      
        \vspace{-2mm}
    \end{subfigure}
    \\
    \Description{
    The figure shows two stacked bar charts representing participants' perceived usefulness of different functionalities across two workflows.
Chart 1: Workflow I – Build Data Science Models
The chart on the left shows how participants rated the usefulness of various functionalities while building data science models:
1. Browse exec history: Most participants (around 8) found it "very useful," with a few finding it "useful."
2. Browse var info: Around 7 participants found it "very useful," and about 2 rated it as "useful."
3. Checkout only data: Most participants rated it as "neutral," with fewer ratings as "less useful" or "useful."
4. Checkout code + data: Participants found it mixed, with most being "useful" and some "very useful" or "neutral."
5. Locate by exec counter: More participants found it "very useful" or "useful."
Locate by search: A similar trend, with more participants considering it "very useful" or "useful."
Chart 2: Workflow II – Recover Work from Bugs and System Crashes
The chart on the right shows participants' ratings for the same functionalities when recovering work from bugs and system crashes:
1. Browse exec history: Similar to Workflow I, most participants found it "very useful."
2. Browse var info: The usefulness rating is more varied, with participants evenly split across "very useful," "useful," "neutral," and "less useful."
3. Checkout only data: Most participants found it "neutral" or "less useful."
4. Checkout code + data: The distribution is similar to "Browse var info," with a few more participants finding it "useful."
5. Locate by exec counter: The majority of participants found it "very useful" or "useful."
Locate by search: It shows a broad distribution, but most participants rated it "very useful" or "useful."
    }

    \vspace{-2mm}
    \caption{Participants' perceived usefulness of each feature across different workflows. All the features are considered "very useful" or "useful" by most participants.}
    \label{fig:feature_usefulness}
\end{figure*}
\begin{figure*}[th]
    \centering       
    \includegraphics[width=0.75\textwidth]{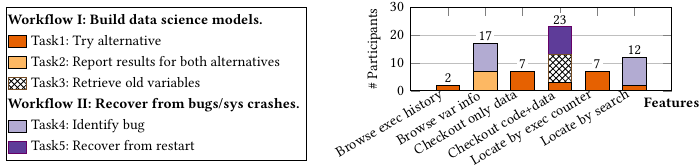}

    \vspace{-4mm}
    \Description{The chart shows that different features were utilized by varying numbers of participants depending on the specific task within each workflow. For example, "Locate by exec counter" was the most frequently used feature, particularly for Task 3: Retrieve old variables, while "Browse var info" was primarily used for Task 4: Identify bug.}
    
    \caption{
    How frequently our features were used.
    All features are used, with different tasks focusing on different ones.}
    \label{fig:usage_frequency}
\end{figure*}

\subsubsection{Participants' direct feedback was overall positive}

We asked participants (in KL and KS) 
    to rate the usefulness of each \system's feature.
The results are reported in Fig~\ref{fig:feature_usefulness}. 
Most users found both ``checkout only data'' and ``checkout both code and data'' to be ``very useful'' when building data science models. 
This is consistent with their behaviors, as we describe shortly. 
Also, most users thought that ``searching'' and ``using execution counter to locate commit'' were useful for finding a commit. 
Participants also mentioned as additional notes that checkout could be useful in their real use cases. 
One person wrote: when they ``want to try out different ways'' or ``find myself executed some wrong code messing up the variable states'', checkout helps them to ``not losing data and code from the previous model'' and ``not need to re-execute from scratch.'' 

Search was marked by 7 users as ``very useful'' for the debugging workflow (second notebook), 
    where they needed to understand how variables evolve. 
According to our observation, when users knew the execution counter for their target commit
    (i.e., an increasing sequence number assigned to each executed cell to indicate
        completed cells), 
    they could locate it more easily (Task 1).
Otherwise, they had to search for the commit they wanted using some criteria (e.g., in Task 4, search by variable changes). 
In addition, a user noted: ``in my use of Jupyter notebooks it becomes hard to find a particular cell in a large notebook, in particular when I use Jupyter notebooks for running code generation models in Google [Colab]. Searching for a particular cell would be my best real use case.''

\subsubsection{\system's features were all frequently used}

In addition to directly asking for ratings,
    we also examined the actual uses of various features offered by \system.
For this test, we manually revisited the screen recordings 
    of participant interactions with \system.
Fig~\ref{fig:usage_frequency} shows how many times each feature was explicitly used
    by the participants.
Seven of the ten participants from the \system groups used both ``checkout past data'' and ``checkout both code and data'' in the ``Build data science models'' workflow. 
Specifically, most participants used "checkout only data" for Task 1
    to start a new exploratory branch while keeping the code.
Everyone used ``check out both code+data'' for Task 3 to hop between branches. 
This result confirms that \system offers the types of navigation that could not be achieved by
    one-dimensional versioning,
    and that different kinds of two-dimensional navigation were employed
        freely by participants to boost their performance.

\subsection{Our proposed UI design is easy to use}
\label{sec:eval-ui_easiness}

We asked the participants about how easy it was to use each UI component (Fig~\ref{fig:eval-ui_easiness}). 

First, most features were rated ``very easy'' or ``easy'' to use from at least 80\% of the users. \fix{For example, although in the experiment, each user created 20 commits on Workflow I and 34 commits on Workflow II on average, 9/10 think it's easy to locate the commit they need, which shows the scalability of the UI design.}
Second, 60\% of the respondents felt ``neutral'' or ``hard'' to understand the difference 
    between dashed lines and solid lines,
        which represent code and data parents in the history graph.
We believe we can largely address this usability issue
    by better highlighting in our tutorial how notebook states may proceed.

\fix{

}

Some participants also commented on potential UI improvements, 
    especially in the search bar, display, and integration.
One respondent suggested extending the search query language 
    to support deeper semantics such as attributes/fields of variables.
Another participant suggests integrating \system into Jupyter.

\begin{figure*}[t]
    \centering
    \begin{subfigure}[b]{1.0\linewidth}
        \centering
            \includegraphics[width=0.75\textwidth]{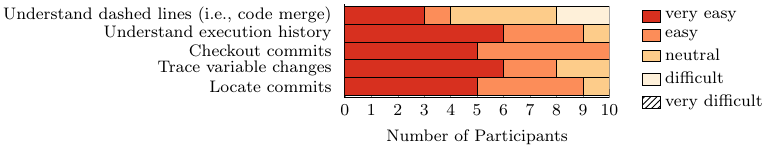}
        \vspace{-2mm}
    \end{subfigure}
    \vspace{-4mm}
    \caption{Participants' perceived easiness to use/understand UI design. Except for the "dashed line," all the other designs are considered "easy" or "very easy" to use/understand by most users.}
    \Description{
    The figure shows a stacked bar chart depicting participants' perceived easiness of using or understanding various aspects of a user interface (UI) design. Except for the "dashed line," all the other designs are considered "easy" or "very easy" to use/understand by most users
    }
    \label{fig:eval-ui_easiness}
\end{figure*}

\subsection{People are willing to use or even pay}
\label{sec:eval-willingness}

\begin{figure*}[t]
    \centering
    \begin{subfigure}[b]{.48\textwidth}
        \centering
            \includegraphics[width=0.9\textwidth]{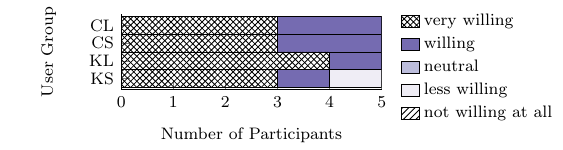}
        \vspace{0mm}
        \caption{Participants’ willingness to \textbf{USE} \system}
        \label{fig:willingness_use}
    \end{subfigure}
    \;\;
    \begin{subfigure}[b]
    {.37\textwidth}
                    \includegraphics[width=0.8\textwidth]{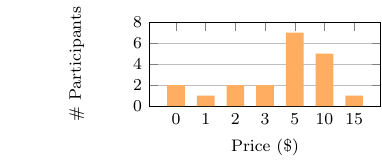}
        \vspace{0mm}
        \caption{Participants’ willingness to \textbf{PAY} for \system}
        \label{fig:willingness_pay}
    \end{subfigure}
    \Description{
    The charts illustrate that most participants are willing to use \system, particularly within the control groups. Additionally, the majority of participants are willing to pay up to \$5, indicating a general positive reception towards the potential pricing for \system.
    }

    \vspace{-2mm}
    \caption{User's willingness to use and pay for \system.  
    The figure shows most participants are willing to use and pay for \system 
        as an extension to Google Colab.}
    \label{fig:willingness}
\end{figure*}

We evaluated \system's overall usefulness by asking
    two types of questions: are you willing to \textbf{use} it in the future, and
        how much would you \textbf{pay} to use it?
        
Fig~\ref{fig:willingness}a presents users' willingness to use \system for their data science tasks. 
The results indicate that most participants were ``very willing'' and ``willing'' to use \system. 
One participant from the KS group said, ``I took the 'Machine Learning' last semester, in the MPs, I need to restart and re-execute to get over-written variables frequently. I wish I had this tool.'' 
Another participant from the CS group said, ``I used to be a data scientist in the industry, and I would say this tool definitely helps my work, as we need to bounce back and forth a lot.'' 

We also asked participants how much they would be willing to pay for \system 
    if, for example, it was a feature of Google Colab+ (Fig~\ref{fig:willingness}b). 
For those who gave a price range as their answer (e.g., \$5-\$10), we report the smaller amount (e.g., \$5). 
The most common choice was \$5/month, with 7 participants indicating the amount. 
Interesting, there were 6 participants willing to pay at least \$10/month for \system. 
Among the two participants who weren't willing to pay for \system, one explained 
``I would not be willing to pay. I already subscribe to Google Colab Pro, and 
    if \system was added as a premium feature it would entice me to continue paying the subscription fee.'' 
From this comment,
    we believe the participant's unwillingness could be partly attributed 
        to the upselling pricing scheme rather than a lack of inherent value in \system.
The particular product name ``Google Colab'' was mentioned in our user study
    to aid understanding;
    however, this project is not affiliated with or sponsored by any for-profit organization.

\fix{With the experiment result, we acknowledge potential biases in self-reported metrics, but we also note that our work offers well-regarded "checkpointing" (e.g., code version control) to a novel domain of data science by noticing the needs of two-dimensional code+data versioning --- a feature well-received by users. }




\section{Related Work}
\label{sec:related}


\begin{table*}[th]

\centering
\caption{Comparison between \system and other systems for supporting non-linear exploration.}
\label{difference_table}

\vspace{-2mm}
\small
\begin{tabular}[H]{lp{80mm}}
\toprule
\textbf{Approach} & \textbf{Mechanism}                               \\ \midrule
Notebook version control system~\cite{head2019managing,kery2019towards,kery2018interactions,gitbook} & Record code version or execution order (cannot recover data). \\
Branched presentation of cells~\cite{ramasamy2023visualising,subramanian2019supporting,venkatesan2022automatic,wang2022stickyland,kery2017variolite} & Present branched structure of code (cannot recover data).\\
Fork kernel~\cite{weinman2021fork} & Fork kernel for alternatives (not scalable, cannot recover past data). \\
Record variable change history~\cite{wang2022diff,hohman2020understanding} & Visualize variables' evolution history (no checkout). \\
Checkpoint and restoration~\cite{chex,elasticnotebook,kishu,txnpython} & 
    Store data science objects/variables (not considering UI/UX). \\
\textbf{Ours (\system)} &
\textbf{Two-dimensional code+data space versioning}
\\
 \bottomrule
\end{tabular}
\end{table*}

This section lists related works in assisting data science workflow on computational notebooks (Table~\ref{difference_table}) \fix{as well as discussing their strengths and drawbacks.}

\paragraph{Notebook Version Control Systems} While Git~\cite{gitbook} can store and manage computational notebooks as source files, many version control systems are specialized for notebook workflows. Verdant~\cite{kery2019towards,kery2018interactions} \fix{records all the operations that happened in the history, like running cell and editing cells, etc., allowing users to browse previous activities and code versions}. \fix{Loops~\cite{eckelt2024loops} allows users to browse and compare the cell code and output for different exploration branches.} Code Gathering~\cite{head2019managing} allows users to automatically extract the smallest necessary code subsets to recreate their chosen outputs. Although these systems help users \fix{understand notebook editing history or} select the relevant code from a potentially messy notebook execution, they do not support recovering data, forcing users to re-execute the code to reach a desired kernel state.
\fix{
    \paragraph{Visualization of Notebook Branches} Another category of previous work helps users to distinguish cells for different branches~\cite{ramasamy2023visualising,wang2022stickyland, harden2023there, harden2022exploring,kery2017variolite}. MARG~\cite{ramasamy2023visualising} designs an index with a branched structure that helps users to quickly locate cells corresponding to the desired branch. Sticky land~\cite{wang2022stickyland} and 2D notebooks~\cite{harden2023there, harden2022exploring} breaks the linear presentation of notebooks and presents cells from different branches in a parallel view. In Variolite~\cite{kery2017variolite}, users can craft alternative code snippets and explore different combinations by selecting different snippet for each step. Like version control methods, these methods help users identify codes relevant to their current exploratory branch; however, to resume the work in a different branch, users still need to re-execute the code. However, from the visualization perspective, these works are orthogonal to \system, we believe they can be combined together so that the users can also see the branched structure of code directly in the notebook itself.
}

\paragraph{Fork Kernel} To support quick data recovery between different branches, Fork it~\cite{weinman2021fork} forks a new notebook for each alternative. However, this method has two problems. First, it is not scalable because forking the kernel per alternative can be expensive.
Second, it does not support forking from a past state; in other words, the user cannot revert the data state to try an alternative, so they must anticipate the forking beforehand.
In contrast, \system implements an incremental checkpointing as discussed in \cite{kishu} and empowers users to explore both past and alternative states.

\paragraph{Record Variable Change History} \cite{wang2022diff,hohman2020understanding} record the variable change history and support comparison between different variable versions. However, they offer no solution to recover the variable to resume notebook execution.
Beyond supporting variable comparison as well as searching for variable changes, \system also allows users to recover variables in their notebook sessions.

\paragraph{Checkpoint and Restoration} Although not the focus of this paper, efficient kernel session checkpoint and restoration techniques are used in our system. Recent advances~\cite{chex,elasticnotebook,kishu,txnpython} enable
    accurate and efficient dirty data identification
    for data science systems.
    The main idea is to identify dirty data
    by extending the standard serialization protocols
        without relying on centralized buffer pools. With this technique, Kishu~\cite{kishu} supports efficient incremental checkpoint and checkout for Python kernel session. 
The approaches can seamlessly integrate with existing notebook systems
        without requiring any intrusive changes. 




\section{Future Work}
\label{sec:future}
Future work should address the following research questions.

\paragraph{How to Support Collaboration of Exploratory Data Science?} Kishu-board is mainly designed for individual use; however, data scientists often synchronously edit the same notebook when collaborating with others~\cite{wang2019data,quaranta2022eliciting,li2023understanding}.
Some tools already support code collaboration~\cite{smith2021meeting} or online conversation during collaboration~\cite{wang2020callisto}. However, merging variables from different users remains an open problem.
This new workflow would simplify many use cases like trying a variable training from one branch in another branch.

\paragraph{How to Support Semantic Search of History Commit?} Currently, \system only supports searching previous commits by their attributes. However, sometimes the users would like to search semantically. For example, in a user study done by~\cite{kery2018story}, users have queries with deeper semantics like ``what was the state of my notebook the last time that my plot had a gaussian is the peak?'' Some work~\cite{li2021nbsearch} leverages advanced machine learning models to enable natural language search for cell code; however, it is unclear how to adapt the technique for semantics in both code and data.

\paragraph{How to Compare Alternatives at Scale?}
Comparison between alternatives is more challenging when facing a large number of choices. Currently, \system only supports diff between the two versions. However, \cite{liu2019understanding} that their participants need to tackle many alternatives, ranging from hundreds of alternative graphs to thousands of data alternatives. Future works need to automate these comparisons through techniques like clustering or text mining.



\section{Conclusion}
\label{sec:con}

In this work, we propose \system, a novel version control system for non-linear computational notebook versions where each contains both code and data states. Our goal is to accelerate data science workflow while keeping the UI intuitive. We interview industrial practitioners to suggest useful features. To achieve consistent code+data versioning while allowing separate checkouts, we formally define consistency and build-in safeguard for safe checkouts. To simplify and visualize \fix{code+data versions}, we introduce a ``git-tree-like'' commit history graph \fix{in one-dimensional axis system} equivalent to the original \fix{two-dimensional axis system}.  \system helps users locate relevant commits through automatic folding, search, and diff. We demonstrate that \system can improve users' efficiency in many exploratory data science tasks and that the UI designs are sound through a user study with 20 participants. Our qualitative feedback also shows that 95\% of the users are willing to use \system in their future work.





\bibliographystyle{ACM-Reference-Format}
\bibliography{refs}

\clearpage

\appendix


\section{Additional Features of \system}
\label{sec:others}

\subsection{Automatic Folding}\label{sec:others_fold}

\system generates a commit for every cell execution—often resulting in hundreds or thousands of commits per session. Displaying all of them at once can hinder browsing, searching, and understanding the overall work history.

To address this, \system uses an auto-folding algorithm (see S4 in Table~\ref{tab:suggestions}) to condense the commit history graph. As shown in Fig~\ref{fig:grouped_commit}, adjacent commits are grouped into a single "grouped commit" node that users can expand by clicking the node (the rectangle with a plus sign).

Auto-folding applies to all commits except for:
\begin{enumerate}
    \item \textbf{User-important commits}: Commits with user-defined tags or messages.
    \item \textbf{Topology-important commits}: Commits that are a leaf (no child), a branch (multiple children), a root (no parent), or a merge (multiple parents).
\end{enumerate}
Fig~\ref{fig:grouped_commit} illustrates which commits are folded and which remain visible, along with the reasons why.



\subsection{Commit Searching}\label{sec:others_search}
\begin{figure}[H]
    \includegraphics[width=.2\textwidth]{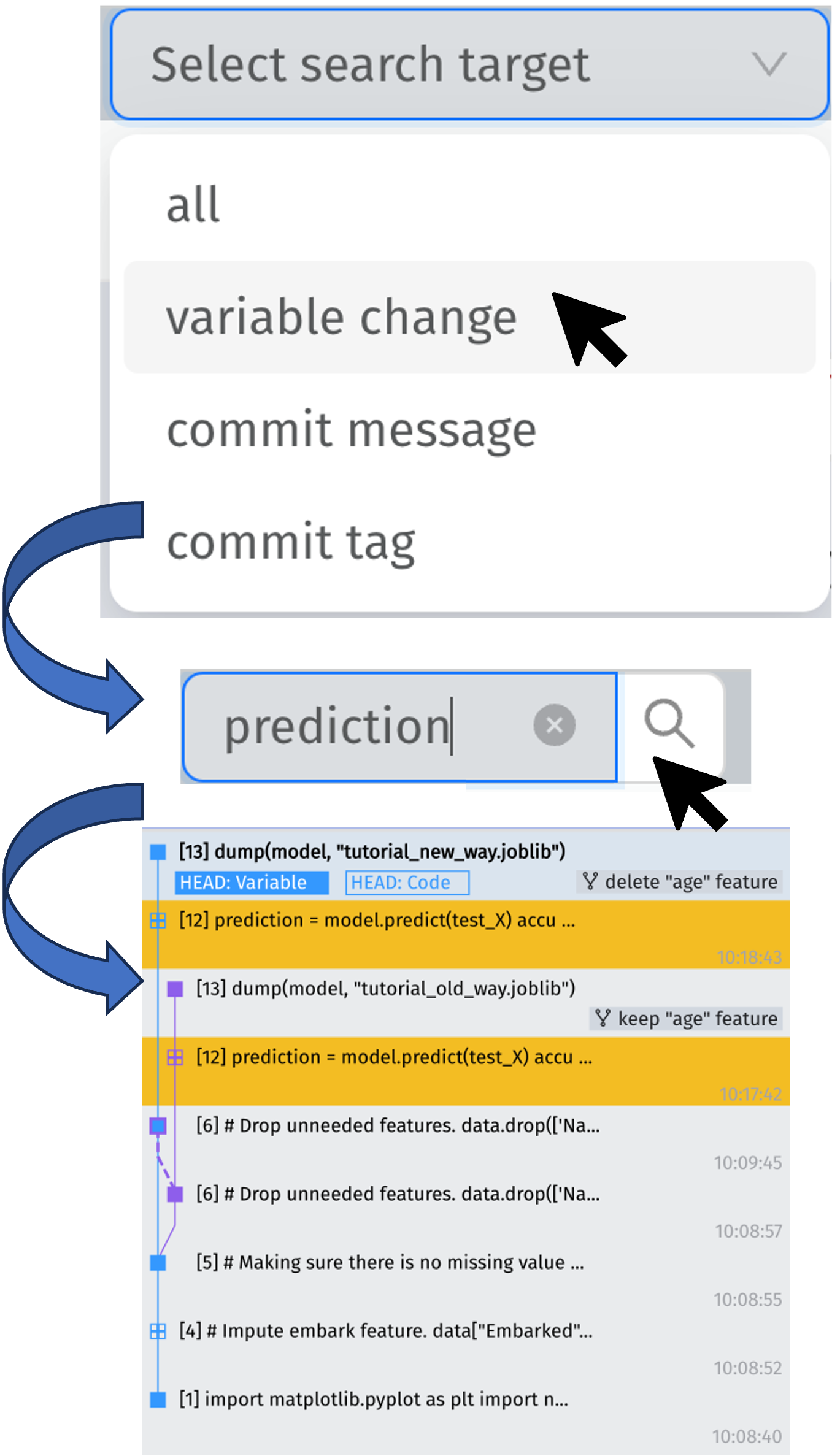}
    \label{fig:search}
    \caption{\system’s commit searching helps users find commits based on their attributes. Users can search for commits according to the commit message and tag. They can also search to highlight all the commits that change a specific variable's value. The figure shows an example of searching all the commits that changed the \textit{``prediction''} variable's value.}
    \Description{The figure demonstrates how to search commits}
\end{figure}


\system's commit searching helps the user find commits based on their attribute. Like other search engines for Git, users can search for commits that have similar tag names and branch names and/or commit messages. Given a search query, \system finds and highlights all matching commits in the commit history graph.

Beyond conventional searches, \system also supports \emph{searches by variable changes} (S5 from Table~\ref{tab:suggestions}). Given a variable name, Kishu-board finds all the commits that change the variable value.  For example, a user, Alice, may wish to understand how a training dataset \texttt{train\_df} has been preprocessed. Through \system, Alice can search by \texttt{train\_df} variable changes and find out about different variants of executed code cells that load the dataset, fill in missing values, and normalize the dataset across different experimental branches.




\subsection{Commit Diff}
\label{sec:others_diff}
\begin{figure}[H]
    \centering
    \begin{subfigure}[b]{\linewidth}
        \centering
        \begin{tabular}{cc}
            \toprule
            Variable Name & Latest Change At \\
            \midrule
            \texttt{data\_df} & $V_1$ \\
            \texttt{model} & $V_2$ \\
            \texttt{fig} & $V_3$ \\
            \bottomrule
        \end{tabular}
        \caption{Variable Version Table at commit $V_3$.}
        \label{fig:var_version_1}
    \end{subfigure}
    \begin{subfigure}[b]{\linewidth}
        \centering
        \begin{tabular}{cc}
            \toprule
            Variable Name & Latest Change At \\
            \midrule
            \texttt{data\_df} & $V_1$ \\
            \texttt{model} & $V_4$ \\
            \texttt{app} & $V_4$ \\
            \bottomrule
        \end{tabular}
        \caption{Variable Version Table at commit $V_4$.}
        \label{fig:var_version_2}
    \end{subfigure}
    \caption{Example variable version tables. Our system quickly diffs variables between any two commits by comparing their corresponding tables.}
    \Description{The figure demonstrates how to detect variable diff.}
    \label{fig:var_version}
\end{figure}

\system's commit diff compares the differences between two commits not only in terms of code but also data states.
Recall that, unlike source code files, a notebook contains a sequence of cells, which are strings of code. To compare code states between two commits, we adapt the existing method specialized for notebooks \cite{nbdime}, which detects both across-cell differences (e.g., added or deleted cells) and within-cell differences (e.g., changed lines of code).

Meanwhile, \system maintains an auxiliary structure called \textbf{variable version table} that keeps track of the variable version for each commit. Each entry of a variable version table is a variable name and the latest commit that changes it. Fig~\ref{fig:var_version_1} shows an example of a variable version table at commit $V_3$. Consequently, \system can quickly identify the difference based on the variable versions without reconstructing the variable values. For example, if a user compares the data state between $V_3$ and $V_4$, \system can lookup the two variable version tables and determine that \texttt{model} is different between the two commits, \texttt{fig} is deleted in $V_4$, and \texttt{app} is added in $V_4$. 

\begin{figure*}
    \centering
    \includegraphics[width=0.75\textwidth]{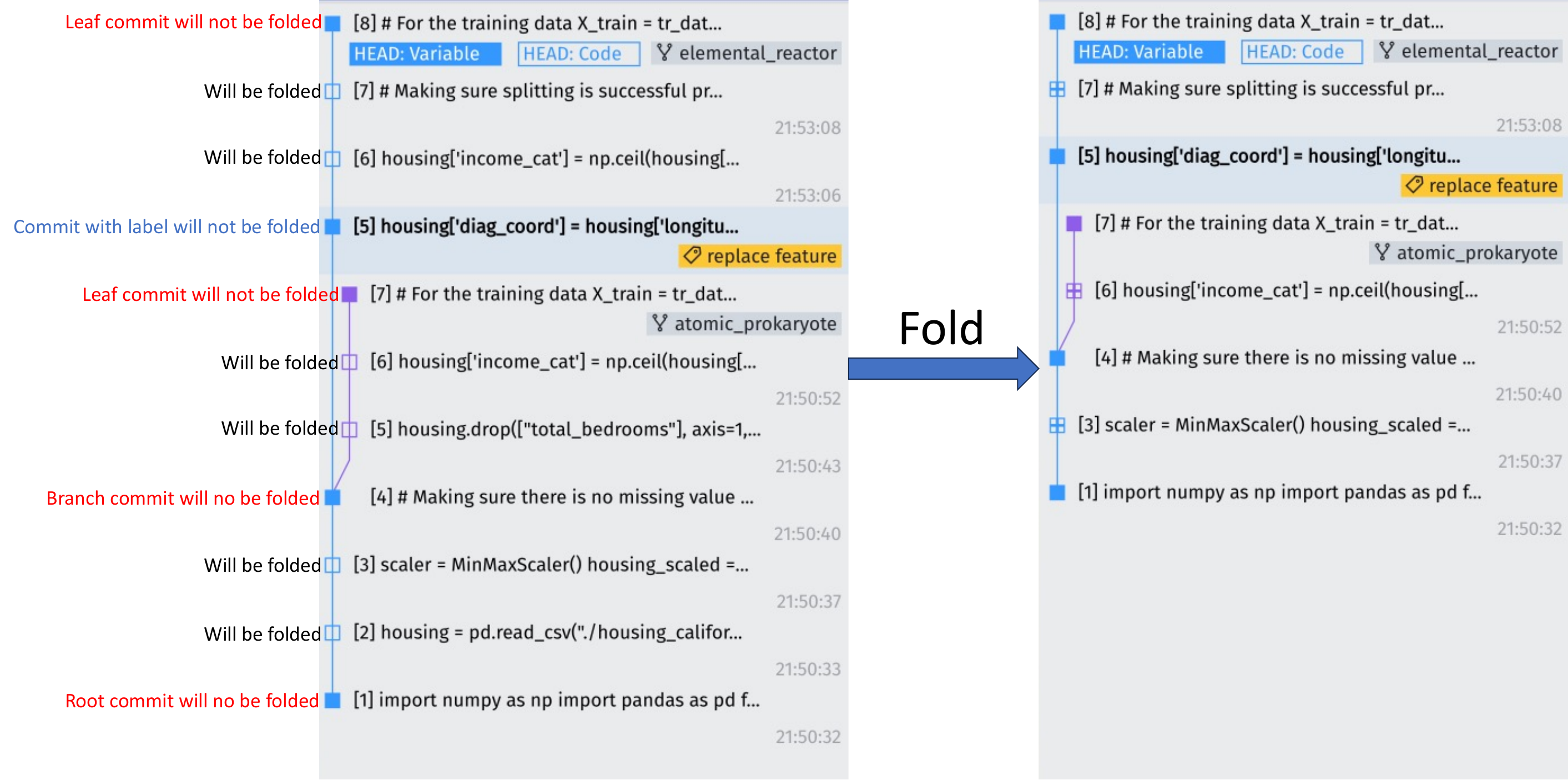}
    \caption{Demonstration of how auto folding works in commit history graph. Other than commits with topology importance (e.g., leaf commit) as indicated in red annotations and user importance (e.g., tag) as indicated in blue annotations, all the other commits can be folded into grouped commits.}
    \Description{
The image shows two panels:
Left Panel: A detailed commit history with several commits listed sequentially, no commit is folded. Each commit is accompanied by a small icon indicating whether or not they'll be folded.
Right Panel: A folded view.
    }
    \label{fig:grouped_commit}
\end{figure*}

\begin{figure*}
    \centering
    \includegraphics[width=0.9\textwidth]{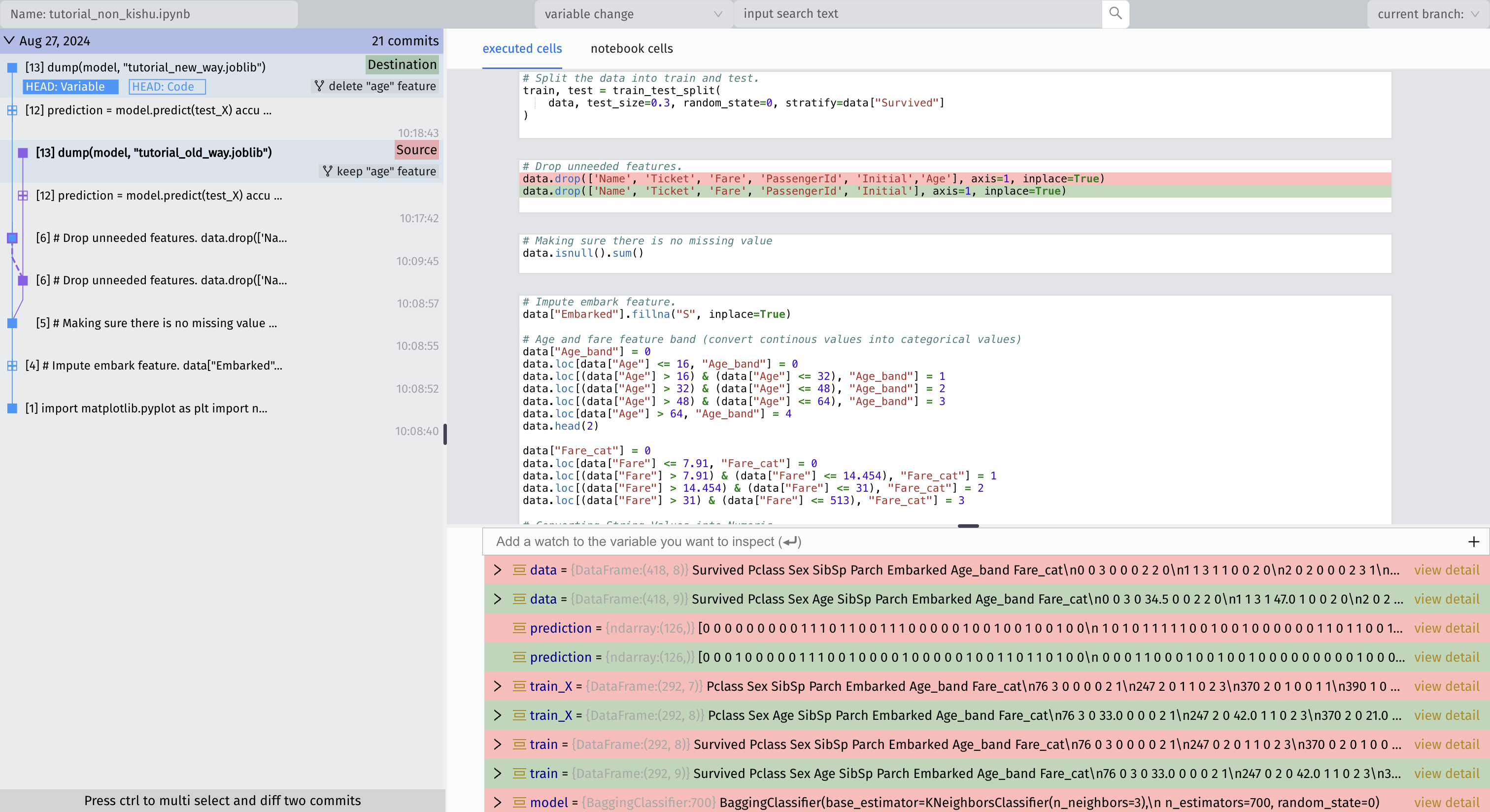}
    \caption{The figure demonstrates how to perform diffs between any two commits.}
    \label{fig:diff}
\end{figure*}



\end{document}